\documentclass[12pt]{article}
\pdfoutput=1

\usepackage{draft}
\usepackage{putex} 
\usepackage[all,cmtip]{xy}
\usepackage{fancybox}
 \usepackage{url}

\usepackage[colorlinks=true]{hyperref}
 \usepackage{extarrows}
 
\usepackage{graphicx,subfig}
\usepackage{cite}
\usepackage{mciteplus}
\usepackage{skak}
\usepackage{amsmath}
\DeclareFontFamily{OT1}{pzc}{}
\DeclareFontShape{OT1}{pzc}{m}{it}{<-> s * [1.10] pzcmi7t}{}
\DeclareMathAlphabet{\mathpzc}{OT1}{pzc}{m}{it}

\def\({\left(}
\def\){\right)}

\newcommand{\CC}{\Gamma}
\newcommand{\pa}{\partial}
\newcommand{\bG}{{\bf G}}

\newcommand{\ep}{\epsilon}

\renewcommand{\cH}{\mathcal{H}}
\newcommand{\eeq}{\end{equation}}
\newcommand{\ea}{\end{array}}

\def\eea{\end{eqnarray}}

\def\<{\langle}
\def\>{\rangle}

\def\bZ{\mathbb{Z}}

\def\cL{\mathcal{L}}
\def\cM{\mathcal{M}}
\def\cN{\mathcal{N}}
\def\cA{\mathcal{A}}

\def\cT{\mathcal{T}}
\usepackage{amsmath}
\usepackage{comment}
\usepackage{amssymb}
\usepackage{amsthm}
\usepackage{bbm}

\newtheoremstyle{dotless}{}{}{\itshape}{}{\bfseries}{}{ }{}
\theoremstyle{dotless}
\newtheorem{thm}{Theorem}
\newtheorem{ax}{Axiom}
\newtheorem{lemma}{Lemma}

\newtheorem{conj}{Conjecture}
\newtheorem{cor}{Corollary}
\theoremstyle{definition}

\makeatletter
\newcommand{\neutralize}[1]{\expandafter\let\csname c@#1\endcsname\count@}
\makeatother

\newenvironment{thmbis}[1]
  {\renewcommand{\theax}{\ref*{#1}$'$}%
   \neutralize{ax}\phantomsection
   \begin{ax}}
  {\end{ax}}
  
\usepackage{color}
\usepackage{tikz-cd}
\usepackage{comment}

\newcommand{\cI}{\mathcal{I}}

\usepackage{adjustbox}

\def\[#1\]{%
  \begin{equation}#1\end{equation}%
}

\usepackage{fancybox}
\usepackage{cite}
\usepackage{framed}
\definecolor{shadecolor}{rgb}{0.9,0.9,0.95}
\definecolor{refkey}{rgb}{0.5,0.5,0}
\definecolor{labelkey}{rgb}{0.5,0.5,0}
\definecolor{citekey}{rgb}{0.5,0.5,0}
\definecolor{darkgreen}{rgb}{0,0.5,0}
\definecolor{darkblue}{cmyk}{0.9,0.9,0,0}
\definecolor{darkred}{rgb}{0.6,0,0.3}

\begin{document}
	
	\preprint{}

	\institution{CMSA}{Center of Mathematical Sciences and Applications, Harvard University, Cambridge, MA 02138, USA}
	\institution{HU}{Jefferson Physical Laboratory, Harvard University,
		Cambridge, MA 02138, USA}

	\title{
    \huge  Anomalous Symmetries End at the Boundary
	}

	\authors{Ryan Thorngren\worksat{\CMSA}    and Yifan Wang\worksat{\CMSA,\HU}}

	\abstract{ 
	A global symmetry of a quantum field theory is said to have an 't Hooft anomaly if it cannot be promoted to a local symmetry of a gauged theory. In this paper, we show that the anomaly is also an obstruction to defining symmetric boundary conditions. This   applies to Lorentz symmetries with gravitational anomalies as well. For theories with perturbative anomalies, we demonstrate the obstruction by analyzing the Wess-Zumino consistency conditions and current Ward identities in the presence of a boundary. We then recast the problem in terms of symmetry defects and find the same conclusions for anomalies of discrete and orientation-reversing global symmetries, up to the conjecture that global gravitational anomalies, which may not be associated with any diffeomorphism symmetry, also forbid the existence of boundary conditions. This conjecture holds for known gravitational anomalies in $D \le 3$ which allows us to conclude the obstruction result for $D \le 4$.}

	\date{}

	\maketitle
	
	\tableofcontents
	
	\section{Introduction}
	
	Given a quantum system with a global symmetry $G$, the 't Hooft anomaly (henceforth simply ``anomaly") is an invariant which represents the obstruction to promoting $G$ to a local symmetry, or equivalently coupling $G$ to background gauge fields. Anomalies are important for quantum field theory because they are preserved under renormalization group (RG) flows of $G$-symmetric theories \cite{tHoof}. This provides us with a litmus test to see if two $G$-symmetric fixed points are connected by a $G$-symmetric flow.
	
	One canonical argument for anomaly-matching employs anomaly in-flow, the observation that for many known anomalies, there is a $G$-symmetric invertible phase in one higher dimension,\footnote{A ($G$-symmetric) invertible phase is a theory $T$ with an inverse $\cT^{-1}$ such that the ``stack" $\cT \otimes \cT^{-1}$ (with the diagonal $G$ action) is equivalent to a trivial theory (with trivial $G$ action). We note that for some gravitational anomalies, e.g. of chiral theories with multiple conformal blocks, there is no invertible bulk phase which makes the bulk-boundary theory invariant. Instead one needs a theory with anyons to represent the different conformal blocks, and such a theory is not invertible. 
	} for which the anomalous theory defines a symmetric boundary condition, such that the combined bulk-boundary system can be coupled consistently to a background gauge field \cite{CallanHarvey}. The anomaly can thus be identified with this invertible phase, also known as the \textit{anomaly field theory}. Anomaly-matching then follows because the RG flows of interest are boundary RG flows for the combined system, and cannot affect the bulk fixed point (which is in fact topological here). This picture is extremely useful also because a classification exists for $G$-symmetric invertible phases, in terms of the cobordism invariants of the spacetime manifold \cite{kapustin2014symmetry,KTTW,freed2019reflection}.
	
	This picture of an anomalous theory as a boundary apparently does not work if we want to consider boundary conditions of the anomalous theory itself, since that would be a ``boundary of a boundary". A similar issue arises in lattice systems, where anomalous symmetries cannot be realized by tensor product operators, and instead must be realized by quantum circuits or evolution by a local Hamiltonian, meaning some arbitrary choices must be made to even define the global symmetry action itself in the presence of a boundary \cite{Wen_2013}. These issues have lead to a kind of folklore in the subject that anomalous symmetries are problematic at a boundary.
	
	In this paper, we prove with mild assumptions (although without invoking anomaly in-flow) that at any boundary of a quantum field theory (QFT), all anomalous symmetries must be broken (either explicitly or spontaneously) at least to a subgroup which is anomaly-free.
	
    For the well-known perturbative (a.k.a local) anomalies such as the chiral anomaly in 1+1D, the reason is intuitively clear: a purely left-moving current cannot be conserved at a boundary because charge cannot flow through the boundary. We formalize this argument and extend it to all dimensions  in Section \ref{sec:conformal}, building upon previous results in \cite{Jensen:2017eof}. By analyzing the Wess-Zumino consistency conditions and the anomaly-descent procedure, we show that the existence of a symmetric boundary requires the corresponding Schwinger term in the descent equations to trivialize, which in turn demands the anomaly polynomial for the relevant symmetries to take a factorized form depending on central $U(1)$ factors of the symmetry group. Furthermore, by analyzing the current Ward identities in the presence of a symmetric boundary in the conformal limit, we prove an obstruction theorem which states such anomalies must vanish for unitary theories.
        This argument also applies to systems with local gravitational anomalies, which shows they cannot have any boundary conditions without breaking the boundary Lorentz symmetry. We also comment on symmetric boundaries for non-unitary theories which circumvent our obstruction theorem. 
    
    We then consider more general global anomalies including those which cannot be diagnosed by the divergence of currents and anomalies of discrete symmetries with no Noether current at all. For these cases we adopt the picture of a global symmetry as a collection of special defects possessing group algebra fusion rules and obeying a list of axioms \cite{Gaiotto:2014kfa}, which we review in Section \ref{sec:defect}. This formulation can be viewed as a generalization of Noether's theorem. The Wess-Zumino consistency conditions follow immediately from these axioms. When these defects are topological (i.e. when there are no gauge-gravity anomalies\footnote{In this paper all gauge fields including the metric are non-dynamical backgrounds for global and spacetime symmetries. As is customary in the literature, we still refer to the relevant 't Hooft anomalies as gauge and gravity anomalies. There are also mixed anomalies that depend nontrivially on both the metric and the background gauge fields, which we refer to as gauge-gravity anomalies.}), these conditions imply that the anomaly is described by a  group cohomology class of $G$. In Section \ref{sec:groupcoh} beginning from a symmetric boundary condition we show that the corresponding cohomology class restricted to the subgroup preserved at the boundary is exact.
    
    More general anomalies (such as gauge-gravity anomalies) must be treated with care, which requires relaxing our conditions on the symmetry defects to allow mild metric and tangent structure dependence, which we describe in Section \ref{sec:gganom}. It is conjectured that the solution to the Wess-Zumino consistency conditions for the anomaly describes a class in a certain cobordism cohomology of the symmetry group $G$. For our arguments it is only necessary that it describes a class in some generalized cohomology theory. We have tried to achieve a balance between physical intuition and mathematical precision in the description of these consistency conditions. We also comment on  phenomena such as group cohomology anomalies becoming trivial in this more general classification.
    
    In Section \ref{sec:genanom} we extend our arguments to boundaries in this general framework, showing once again that no symmetric boundary condition exists. This argument relies on the conjecture that systems with gravitational anomalies cannot have boundaries. We have proven this for perturbative anomalies of Lorentz invariant theories in Section~\ref{sec:conformal}, and in Section \ref{sec:genanom} we are able to show it for enough global anomalies that we can conclude the main result for spacetime dimensions $D \le 7$ for fermions and $D \le 4$ for bosons. To extend this result will require a better understanding of the general gravitational anomalies in higher dimensions.
    
    In Section \ref{sec:ainf} we comment on anomaly in-flow, and prove that all symmetries, as we have defined them, satisfy anomaly in-flow.
       Finally in Section \ref{sec:discussion} we comment on extending our results to higher form symmetries, domain walls, and consequences for emergent anomalies and gauge theories in 3+1D.
    
    \section{Perturbative Anomalies and Boundaries}
    \label{sec:conformal}
   Here we discuss the interplay between perturbative anomalies of a Lorentz invariant QFT $\cT$ with global symmetry $G$ in even spacetime dimensions $D=2n$ and its possible Lorentz invariant boundary conditions $\cB$. We will show that a continuous global symmetry subgroup $G_\cB \subset G$ can be preserved at the boundary only if it has a trivial 't Hooft anomaly (including gauge-gravity anomalies whose anomaly polynomial involves the Riemann curvature as well as gauge curvatures). Moreover, $\cT$ will not admit any such boundary conditions if there is a perturbative pure gravitational anomaly.

\subsection{Review of Perturbative Anomalies}
Let us first briefly review the perturbative anomalies of QFTs in the absence of boundaries.\footnote{See \cite{Bertlmann:1996xk} for a comprehensive review including the mathematical background.} We denote collectively the continuous global symmetry $G$ and Lorentz symmetry of the theory $\cT$ by
\ie
\cG=G\times SO(2n)\,.
\fe
Upon coupling the theory to background $\cG$ gauge fields $B$, the perturbative 't Hooft anomalies manifest through the anomalous variation of the partition function,
\ie
\cA(v,B) \equiv -i\D_v \log Z[B]\,,
\label{aV}
\fe
under a gauge transformation parametrized by $v$ with $\D_v B=dv+[B,v]$. The anomalous variation satisfies the Wess-Zumino (WZ) consistency condition \cite{Wess:1971yu}
\ie
{}\D_{v_1} \cA(v_2,B)-\D_{v_2} \cA(v_1,B)=\cA([v_1,v_2],B)\,.
\label{WZ}
\fe
This equation ensures that the infinitesimal gauge transformations integrate to an action of the group of gauge transformations. Solutions to this equation are given by the Stora-Zumino descent procedure \cite{Zumino:1983ew,AlvarezGaume:1984dr,Manes:1985df}. 

We follow \cite{AlvarezGaume:1984dr} here. Let $\theta^\alpha$ be a set of parameters parametrizing a family of gauge transformations $g(x,\theta)$ with $g(x,0) = 1$. We define the corresponding family of transformed background gauge fields
\ie
\bar B(x,\theta)=g^{-1} (B+d)g,
\label{hB}
\fe
which satisfies $\bar B(x,0)= B(x)$. We define the exterior derivative in the parameter directions
\ie
\hat d \equiv d\theta^\A  {\pa \over \pa \theta^\A}\,.
\label{edt}
\fe
The infinitesimal gauge parameters are given by the Maurer–Cartan one-form of $\cG$,
\ie
\hat v \equiv v_\A d\theta^\A  = g^{-1}  \hat d  g\,,
\label{MC1}
\fe
which satisfies
\ie
\hat d \hat v =-\hat v\wedge \hat v\,.
\fe
From \eqref{hB} and \eqref{MC1} we see $\hat d$ acts on $\bar B$ and its curvature $ F(\bar B)$ as a gauge transformation with parameter $v_\A$
\ie
\hat d \bar B =-D_{\bar B} \hat v
,\quad 
\hat d  F (\bar B)=[ F(\bar B), \hat v]\,.
\fe
The Wess-Zumino consistency condition \eqref{WZ} follows from
\ie
 \hat d \cA(\hat v, \bar B) =0\,.
 \label{WZg}
\fe
Indeed if we choose $g(x,\theta)$ such that $\left.v_\A(x,\theta)\right |_{\theta=0}$ coincide with the gauge parameters in   \eqref{WZ} for $\A=1,2$, equation  \eqref{WZ}  comes from the coefficient of $d\theta^1\wedge d\theta^2$ at $\theta=0$ from the above expression.

Locality requires $\cA(v,B)$ to be written as an integral of a density $\int_\cM Q(v,B)$ on the spacetime manifold $\cM$, and likewise $\cA(\hat v, \bar B) = \int_\cM Q(\hat v, \bar B)$. The equation \eqref{WZg} is then equivalent to $\hat d Q$ being a total derivative on $\cM$. Solutions of this form come in a sequence of differential forms satisfying the descent equations
\ie
\hat d Q_{2n+1}^{(0)}+dQ_{2n}^{(1)}=&0\,,
\\
\hat d  Q_{2n}^{(1)}+dQ_{2n-1}^{(2)}=&0\,,
\\
\vdots~~~&
\\
\hat d Q_{1}^{(2n)}+dQ_{0}^{(2n+1)}=&0\,,
\\
\hat d  Q_0^{(2n+1)}=&0\,,
\label{SZd}
\fe
where $Q_m^{(k)}$ is a degree $m$ polynomial in $\bar B$ (which has degree 1) and its field strength $F(\bar  B)$ (which has degree 2), and has degree $k$ in the gauge parameter $\hat v$. At the top of the descent equations is $Q_{2n+1}^{(0)}$, a Chern-Simons-type term which represents the action of a $D+1$-dimensional bulk theory on the boundary of which $\cT$ is gauge invariant. It is associated with a degree $2n+2$ anomaly polynomial $\cI_{2n+2}[\cT] = d Q_{2n+1}^{(0)}$, a polynomial in the background gauge field strength $F(\bar B)$ (which includes the Riemann curvature 2-form $R$). Note that the terms $Q_{2n+1-k}^{(k)}$ in \eqref{SZd} are subjected to ambiguities of the form 
\ie
Q_{2n+1}^{(0)} \to Q_{2n+1}^{(0)} + d\A_{2n}^{(0)},~ Q_{2n}^{(1)} \to Q_{2n}^{(1)} + \hat d \A_{2n}^{(0)}+ d \A_{2n-1}^{(1)},
\dots 
\label{damb}
\fe
but the solutions are physically equivalent \cite{Zumino:1984ws}.

The solution to \eqref{WZg} is readily obtained from the descent equations \eqref{SZd} as
\ie
\cA(\hat v,B)= \int_\cM \left. Q^{(1)}_{2n}\right|_{\theta=0}\,,
\label{WZs}
\fe
where $\cM$ is the closed spacetime manifold.
The rest of the terms in the descent equations \eqref{SZd} also have physical origins. In particular $Q^{(2)}_{2n-1}$ is responsible for the modification of the equal-time commutation relation of the conserved currents acting on the Hilbert space of $\cT$ in the presence of background gauge fields\cite{Zumino:1984ws,Faddeev:1984jp,Faddeev:1985iz,Faddeev:1986pc}. Quantizing the theory on a time-slice $\cS$ of $\cM$, we define the (smeared) Gauss-law operator for $\cG$ as usual 
\ie
{\bf G}(v)=\int_\cS (j^i_0(  \sigma) +X^i( \sigma)) v^i( \sigma)
\fe
where $\sigma^a$ are the coordinates on $\cS$,  $j_0^i$ is the time-component of the corresponding Noether current where $i$ is the adjoint index for $\cG$ and $X^i\equiv -(D_a)^i{}_j {\delta \over \delta B_a^j} $ generates space-dependent gauge transformations of the background gauge field $B$. The commutator of the Gauss-law operators can differ from that of the Lie algebra of $\cG$,\footnote{The classical gauge transformation generators $X^i(\sigma^a)$ obey the undeformed commutation relation.}
\ie
{}[\bG(v_\A),\bG(v_\B)]=\bG([v_\A,v_\B])+\int_\cS  S(v_\A,v_\B,B)
\label{etc}
\fe
where the correction term is known as the (integrated) Schwinger term, which captures the contact term in the equal-time commutator of the conserved currents and equivalently the 
projective representation of the symmetry transformations on the Hilbert space \cite{Zumino:1984ws}.\footnote{The Schwinger (contact) term for $D=1+1$ is independent of background gauge fields as evident from the central extensions of the current algebras. In higher dimensions, this is no longer the case.} 

Similar to how the anomalous variation $\cA$ is constrained by the WZ consistency condition, the Schwinger term is constrained by the Jacobi identity. Up to $\hat d$-exact c-number ambiguities due to redefinitions of the Gauss-law operators by terms involving the background  gauge field, the solution is determined by the anomaly and given by the following term in the descent equations \cite{Zumino:1984ws}
\ie
S(v_\A,v_\B,B)=\left. Q_{2n-1}^{(2)}(\hat v, B) \right |_{d\theta^\A d\theta^\B}\,.
\label{SWt}
\fe
This will be relevant to us later when we include a boundary for the spacetime manifold.

Perturbative anomalies also manifest in the modification of current conservation laws by contact terms. For example, an anomalous symmetry current $J^\m$ in $D=2n$ spacetime dimensions is characterized by the following modification of the current Ward identity,
\ie
\la \pa_\m J^\m(x) J^{\m_1} (x_1)\dots J^{\m_{n}} (x_{n})\ra=-{k\over (n+1)!(2\pi)^n}
\ep^{\m_1 \dots \m_n \n_1 \dots \n_n}
\prod_{i=1}^n {\pa \over \pa x_i^{\n_i} }\D^d(x-x_i).
\label{aJw}
\fe
Note that the RHS is constrained to take the form above so that the anomaly is $U(1)$ invariant, a consequence of the Wess-Zumino consistency condition. When coupled to a background $U(1)$ gauge field $A$, it leads to an anomalous variation of the partition function, under a gauge transformation $A\to A+d\lambda$,
\ie
\D_\lambda \log Z[A]={k i \over (n+1)!(2\pi)^{n+1}}\int_\cM  \lambda F^n\,.
\label{aJt}
\fe
Equivalently, the $U(1)$ anomaly is characterized by a degree $2n+2$ anomaly polynomial,
\ie
\cI_{d+2}={k\over (n+1)!(2\pi)^{n+1}} F^{n+1}
\label{aJp}
\fe
which reproduces the anomaly \eqref{aJt} through the descent equations \eqref{SZd},
\ie
\cI_{2n+2}=d Q^{(0)}_{2n+1},\quad \D_\lambda Q^{(0)}_{2n+1}= dQ^{(1)}_{2n},\quad \D_\lambda \log Z[A]=i\int_{\cM} Q^{(1)}_{2n}\,.
\fe
with $Q^{(1)}_{2n}={k\over (n+1)!(2\pi)^{n+1}} \lambda F^n$ from \eqref{aJp}.
Here $Q^{(0)}_{2n+1}$ is the Chern-Simons $2n+1$-form that realizes the anomaly-inflow from a gapped auxiliary bulk theory in $2n+1$-dimensions to the physical theory on $\cM$.

More generally, through the descent equations \eqref{SZd}, the anomaly polynomial $\cI_{d+2}[\cT]$ determines the anomalous variations under background gauge transformations and local Lorentz rotations parametrized by $\lambda$ and $\theta$ respectively,
\ie
\D_{\lambda,\theta}  \log Z[A,e]=i \D_{\lambda,\theta} \int_{W} Q_{2n+1}^{(0)}(A,\omega ) = i\int_{\cM} Q^{(1)}_{2n}(\lambda,\theta, F,R) \,,
\label{lorvar}
\fe
which solves the Wess-Zumino consistency conditions \cite{Wess:1971yu}. Here $W$ is an auxiliary $2n+1$ dimensional manifold with boundary $\pa W=\cM$. 
In the above we use $e$ to denote the vielbein and $\omega$ is the spin-connection, which transform  under the Lorentz rotation as,
\ie
\D_\theta e_\m^a = -\theta^a{}_b e_\m^b,\quad \D_\theta \omega_\m ^{ab} = \nabla_\m \theta^{ab}  \,.
\fe
As is well known, the Bardeen-Zumino counter-term \cite{Bardeen:1984pm} allows one to shift between Lorentz and diffeomorphism anomalies. In the above, we have implicitly assumed the scheme where the diffeomorphism anomaly vanishes. If instead, we insist on a symmetric stress-tensor in correlation functions including at coincident points, the Lorentz anomaly gets replaced by a diffeomorphism anomaly. Then instead of \eqref{lorvar},  under reparametrization $\D x^\m =\xi^\m(x)$ we have,
\ie
\D_{\lambda,\xi}  \log Z[A,g]=i \D_{\lambda,\xi} \int_{W} \tilde Q_{2n+1}^{(0)}(A,\Gamma) = i\int_{\cM} \tilde Q^{(1)}_{2n}(\lambda,\xi, F,\Gamma) \,,
\label{diffvar}
\fe
where $\CC^\m_{\n\rho}$ is the Christoffel connection and $\tilde Q_{2n+1}$ differs from $Q_{2n+1}$ by an exact $2n+1$-form.

\subsection{Symmetric Boundaries and Vanishing Schwinger Terms}
\label{sec:vanst}

Let us now place the theory $\cT$ with anomaly polynomial $\cI_{2n+2}[\cT]$ on half space $\mR^{2n}_+$  with coordinates $x^\m=(\sigma^a,x^\perp)$
and a putative Lorentz invariant boundary condition $\cB$ at $x^\perp=0$. The coupled system is commonly referred to as a boundary field theory which we denote by $\cB[\cT]$. The symmetry preserved includes the Lorentz subgroup $SO(2n-1)$ acting on the boundary directions and a subgroup of the bulk global symmetry $G_\cB \subset G$, which we denote collectively by $\cG_\cB$. Below we will deduce constraints on the bulk anomaly polynomial $\cI_{d+2}[\cT]$ from the existence of such a boundary condition $\cB$.

We first observe that in the presence of a boundary $\Sigma\equiv \pa \cM$, \eqref{WZs} is not $\hat d$-closed in general and thus the WZ consistency condition is no longer satisfied. Instead one finds using \eqref{SZd}
\ie
\hat d \cA(\hat v,B)=\hat d \int_\cM Q_{2n}^{(1)}(\hat v,B)=-\int_\Sigma Q_{2n-1}^{(2)}(\hat v,B)\,,
\fe
where the background gauge field $B$ is restricted here to the symmetry subgroup $\cG_\cB$ preserved by the boundary, similarly $\theta^\A$ are restricted to be coordinates on $\cG_\cB$ (see around \eqref{edt}).
To fix the WZ consistency condition with a boundary, we need to modify the anomalous variation by boundary contributions \cite{Jensen:2017eof}
\ie
\cA_\cB (\hat v,B)=\cA(\hat v,B) + \int_\Sigma V(\hat v,B)\,,
\fe
such that
\ie
\hat d \cA_\cB (\hat v,B)=0 \,.
\fe
This is only possible if
\ie
\int_\Sigma \hat d V (\hat v,B)=\int_\Sigma Q_{2n-1}^{(2)}(\hat v,B)\,.
\label{Scond}
\fe
By regarding the orthogonal direction to the boundary as the Euclidean time, and taking $\cS=\Sigma$ to be the spatial slice, we see \eqref{Scond} requires the Schwinger term \eqref{SWt} on $\Sigma$ to be trivial, and thus can be set to zero after a redefinition of the charge densities. This is indeed natural in the following sense. The boundary condition $\cB$ corresponds to a particular state $|\cB\ra $ in the Hilbert space on $\Sigma$ (in the presence of background gauge fields). The fact that $\cB$ respects the $\cG_\cB$ symmetry translates to the following condition for the corresponding Gauss-law operator (after a c-number redefinition if necessary) that implements gauge transformations on the boundary parametrized by a Lie algebra valued function $v(\sigma^a)$,
\ie
\bG (v)|\cB\ra =0\,.
\fe
Consistency with the algebra of $\bG(v)$ in \eqref{etc} then demands the Schwinger term to vanish.\footnote{This also means that the symmetry transformations of $\cG_B$ on the Hilbert space on $\Sigma$ cannot be projective.}

The triviality of the Schwinger term in the descent equations places strong constraints on the anomaly $\cI_{2n+2}[\cT]$. Suppose there is an anomaly of the form 
\ie
\cI_{2n+2}[\cT]=P(F(B)^{n+1})\,.
\fe
Here $P(X_1,X_2,\dots X_{n+1})$ denotes a symmetric invariant polynomial of degree $2n+2$ in the Lie algebra valued variables $X_i$ (of degree 2). If some of the $X_i$ are equal, e.g. $X_1=X_2=\dots=X_m=X$, we write compactly
\ie
P(X^m,X_{m+1},\dots X_{n+1})\,.
\fe
The Schwinger term is determined by (up to coboundaries)
\ie
Q_{2n-1}^{(2)}
=
\begin{cases}
P(\hat v, d\hat  v) & n=1
\\
{n(n-1)(n+1)\over 2}\int_0^1 dt (1-t)^2 P((d\hat v )^2,\bar B,F_t(\bar B)^{n-2}) & n \geq 2
\end{cases}
\fe
with $ F_t(\bar B)\equiv t d\bar  B+t^2 \bar B\wedge \bar  B$ as given in \cite{Zumino:1985vr}. To be compatible with \eqref{Scond}, we must have
\ie
\hat d\int_\Sigma Q_{2n-1}^{(2)}(\hat v,B)=0\,,
\fe
since $\hat d^2=0$. A quick inspection reveals that this is not possible unless $(d\hat v)^2=0$ which requires the relevant gauge parameters to be abelian \cite{Jensen:2017eof}. Therefore, pure non-abelian anomalies are not compatible with the WZ consistency conditions. This lead us to the following theorem, which was already argued for in \cite{Jensen:2017eof} and we have re-derived here.
\begin{thm}
A $2n$-dimensional QFT $\cT$ may admit a symmetric boundary condition $\cB$ only if its  anomaly polynomial is a sum of monomials with the factorized form
\ie
\cI_{2n+2}[\cB[\cT]]= \sum_I F^I_{U(1)}\wedge H^I_{2n}(F,R)
\label{APres}
\fe
when restricted to the symmetry subgroup $\cG_\cB\subset \cG$ preserved by the boundary.
\label{thm:WZ}
\end{thm}
In the above $F^I_{U(1)}$ is the field strength of an abelian factor in the center $U(1)_I\subset \cZ(\cG_\cB)$ and $H_{2n}^I$ is a symmetric invariant polynomial of degree $2n$ in the background curvatures (here $F$ may include $F^I_{U(1)}$). The Schwinger term simply vanishes in this case (up to the ambiguities in the descent equations \eqref{damb} as usual),
\ie
Q^{(0)}_{2n+1}= A^I_{U(1)} H_{2n}^I(F,R),\quad 
Q^{(1)}_{2n}= \hat v^I_{U(1)} H_{2n}^I(F,R),\quad 
Q^{(2)}_{2n-1}=0\,,
\fe
where $\hat v^I_{U(1)}$ contains the gauge transformation parameter for the $U(1)_I$ symmetry.

In particular if $\cT$ has a pure gravitational anomaly which is possible for $n\in 2\mZ+1$ \cite{AlvarezGaume:1983ig,AlvarezGaume:1984dr}, it cannot have a Lorentz invariant boundary condition if $D>2$.
\begin{cor} For spacetime dimension $D=2n>2$, if the theory $\cT$ has a pure gravitational anomaly, it cannot admit a Lorentz invariant boundary condition $\cB$ preserving the $SO(2n-1)$ subgroup.\footnote{We don't lose information of the bulk pure gravitational anomalies upon reduction of the structure group from $SO(2n)$ to $SO(2n-1)$ since the relevant Pontryagin classes $p_1,p_2,\dots,p_{n+1\over 2}$ remain independent as long as $n\geq 3$. A similar reasoning applies to gauge-gravity anomalies for $n\geq 2$.}
\label{cor:pertg}
\end{cor}

Note that the gravitational anomaly in $D=2$ (i.e. $n=1$) is not constrained by Theorem~\ref{thm:WZ} since the Lorentz group is completely broken by the boundary. Nevertheless it has been shown that such an anomaly is an obstruction to boundary conditions for 2d theories based on a CFT argument \cite{Jensen:2017eof}. In the next section we will extend this result to higher spacetime dimensions. 

Before we end this section let us comment on a caveat concerning unitarity and anomalies. Thus far we have not demanded the QFT of interest $\cT$ and its boundary condition $\cB$ to be unitary. However we have implicitly assumed that the perturbative anomalies of $\cT$ are all captured by the descent procedure. It is known that more \textit{exotic} perturbative anomalies that solve the WZ consistency conditions are possible in non-unitary theories  \cite{Nakayama:2018dig,Chang:2020aww}. In the rest of the paper, we will take the theory $\cT$ and its boundary $\cB$ to be unitary unless explicitly stated otherwise. We comment on this point further in Section~\ref{sec:nuea} after proving Theorem \ref{thm:vanpa} below.

\subsection{Conformal Boundaries and Vanishing Anomalies}

By studying the Ward identities, we will further demonstrate that $\cI_{d+2}[[\cB[\cT]] = 0$ in this section, where we recall that $\cI_{d+2}[[\cB[\cT]]$ denotes the bulk anomaly polynomial $\cI_{d+2}[\cT]$ restricted to symmetries preserved by the boundary $\cB$.
Since the 't Hooft anomalies are RG invariants, it suffices to focus on the infra-red phase of the boundary field theory $\cB[\cT]$, which is expected to be described by a conformal field theory (CFT) with certain conformal boundary condition, also known as a boundary CFT (BCFT) \cite{Diehl:1981zz,Cardy:1984bb,Cardy:1991tv,McAvity:1993ue,McAvity:1995zd} (see \cite{Cardy:2004hm,Andrei:2018die} for recent reviews).

By assumption there is a global $U(1)$ symmetry preserved by the boundary $\cB[\cT]$. Its Noether current satisfies
\ie
  \pa_\m J^\m(x)   = 0\,,
  \label{bJc}
\fe
everywhere including at the boundary $x^\perp=0$ away from other operator insertions, similarly for the $SO(2n-1)$ Lorentz symmetry parallel to the boundary
\ie
\pa_\m T^{\m a} (x) = 0 \,. 
\label{bTc}
\fe
 Note that in general the conserved current in \eqref{bJc} is a linear combination of bulk and boundary operators,
\ie
J^\m (x)= J_{(0)}^\m(x)+  \sum_{m\geq 1}\D^{(m)}(x^\perp)  J_{(m)}^\m (\sigma)\,,
\label{Jexp}
\fe
where $J_{(m)}^\m$ is a boundary operator of scaling dimension $\Delta=2n-1-m$ which splits as $(J_{(m)}^\perp,J_{(m)}^a)$  into a scalar and a vector under the residual $SO(2n-1)$  Lorentz symmetry. Now unitarity bounds based on the boundary conformal algebra $SO(2n,1)$ implies\footnote{See \cite{Rychkov:2016iqz} for a review on CFT techniques and in particular the conformal unitarity bounds.}
\ie
\Delta (J_{(m)}^\perp) \geq {d-2\over 2},\quad \Delta (J_{(m)}^a) \geq {d-2}\,.
\fe
Thus we conclude the sum in \eqref{Jexp} truncates to
\ie
J^a (x)=& J_{(0)}^a(x)+  \D(x^\perp)  J_{(1)}^a (\sigma)\,,
\\
J^\perp (x)=& J_{(0)}^\perp(x)+   \sum_{m= 1}^{ \lfloor d/ 2 \rfloor}\D^{(m)}(x^\perp)  J_{(m)}^\perp(\sigma)\,,
\fe
where $J_{(1)}^a$ is a locally conserved current on the boundary. Let's consider the integrated Ward identity of the form
\ie
\lim_{\ep \to 0}\int_{-\ep}^\ep dx^\perp (x^\perp)^m \pa_\m J^\m =0  \,.
\fe
With $m\geq 1$, one finds $J_{(m)}^\perp=0$ as $x_\perp \to 0$. Next taking $m=0$, it gives
\ie
\lim_{x_\perp \to 0} J^\perp_{(0)} =-\D(x^\perp) \pa_a J_{(1)}^a =0\,.
\fe
Consequently we have
\ie
\lim_{x_\perp \to 0} J^\perp =0\,,
\label{Jv}
\fe
as an operator identity. Note that this is consistent with the vanishing Schwinger term when we treat the $x_\perp$ direction as the Euclidean time, as discussed in the last section.

Importantly the anomalous Ward identity \eqref{aJw} cannot be modified in the presence of a symmetric boundary. This is because such a modification is equivalent to a parity-odd gauge-invariant density on the boundary, which is not possible in odd dimensions. On the other hand, 
\eqref{Jv} implies
\ie
\lim_{x_n^\perp \to 0} \la J^\m(x) J^{\m_1}(x_1)\dots J^{\perp }(x_n)\ra =0\,,
\fe
and thus
\ie
\lim_{x_n^\perp \to 0} \la \pa_\m J^\m(x) J^{\m_1}(x_1)\dots J^{\perp }(x_n)\ra =0\,,
\fe
which implies $k=0$ in \eqref{aJw} and this rules out pure $U(1)$ anomalies in \eqref{APres}. From the argument leading to \eqref{Jv}, it is clear that this continues to hold for general conserved currents $J_\m^i$ that generate symmetry $G_\cB$. The corresponding anomalous Ward identity takes the form
\ie
\la \pa_\m J^\m(x) J_{i_1}^{\m_1} (x_1)\dots J_{i_n}^{\m_{n}} (x_{n})\ra=-{\kappa \over (2\pi)^n} K_{i_1\dots i_n}
\ep^{\m_1 \dots \m_n \n_1 \dots \n_n}
\prod_{i=1}^n {\pa \over \pa x_i^{\n_i} }\D^d(x-x_i)\,,
\label{}
\fe 
where $K_{i_1\dots i_n}$ is a $G_\cB$-invariant tensor. We conclude $\kappa=0$ by taking the $x_n^\perp \to 0$ limit and using $\lim_{x_n^\perp \to 0} J^\perp_{i_n}(x_n)=0$. Therefore all but mixed $U(1)$-gravitational anomalies in \eqref{APres} are forbidden.

The anomalous Ward identity for such a gauge-gravity anomaly takes the following form for $D=4$,
\ie
\la \pa_\A J^\A(x) T_{\m \n}  (y)  T_{\rho \sigma} (z)\ra=&-{k_g\over (2\pi)^2}  
\ep_{\m \rho}{}^{ \A \B} \pa^y_\A \pa^z_\B(\pa^y\cdot \pa^z \D_{\n \sigma} -\pa^y_\sigma \pa^z_\n)
 \D^4(x-y)\D^4(x-z)
 \\
 &+(\rho \leftrightarrow \sigma)+(\m \leftrightarrow \n)\,.
\label{JTTw}
\fe
 A parallel argument for \eqref{Jv} shows that \eqref{bTc} requires\footnote{The boundary limit of the other component 
 \ie
 \lim_{x_\perp \to 0} T^{\perp\perp} ={\rm D}(\sigma)
 \fe
 defines a boundary operator ${\rm D}(\sigma)$, known as the displacement operator which is nontrivial if the boundary is not topological.
 }
\ie
\lim_{x_\perp \to 0} T^{a\perp} =0\,.
\label{Tv}
\fe
Consequently by taking the limit and using
\ie
\lim_{z^\perp \to 0} \la J^\A(x) T_{\m \n}  (y)  T_{\rho \perp} (z)\ra =0 \,,
\fe
we deduce that the mixed $U(1)$-gravitational anomaly must vanish for $D=4$. A similar argument shows this continues to hold in higher dimensions. This concludes the argument for the following theorem.
\begin{thm}
A unitary QFT $\cT$ in dimension $D=2n$ can admit a unitary Lorentz invariant boundary condition $\cB$ that preserves a global symmetry subgroup $G_\cB$ only if the theory $\cT$ does not have perturbative gravitational and $G_\cB$ anomalies. In particular, the anomaly polynomial must trivialize 
\ie
\cI_{d+2}[\cB[\cT]] =0\,,
\label{pertv}
\fe
when restricted to bulk symmetries $\cG_\cB=SO(2n-1)\times G_\cB$ preserved by the boundary. 
\label{thm:vanpa}
\end{thm}

\subsection{Comments on Non-Unitary Theories and Exotic Anomalies}
\label{sec:nuea}

We emphasize that Theorem~\ref{thm:vanpa} does not apply to non-unitary theories, which may have more general anomalies as mentioned at the end of Section~\ref{sec:vanst}. A familiar counter-example from string theory is the non-unitary (non-chiral) $bc$ ghost CFT in $D=2$ described by the following action (see \cite{Polchinski:1998rq} for details),
\ie
S_{bc}={1\over 2\pi}\int d^2 z (b \pa_{\bar z} c+ \bar b \pa_{z}\bar c)\,,
\fe
on the complex plane with coordinates $(z,\bar z)$ and $z= \sigma +i x_\perp$. Here  $b,c$ are holomorphic anti-commuting fields (ghost) and $\bar b,\bar c$ are their anti-holomorphic partners. The $bc$ CFT is parametrized by a real number $\lambda$ which determines the  holomorphic and anti-holomorphic conformal weights $(h,\bar h)$ of the ghost fields,\footnote{When the $bc$ CFT is placed on a curved manifold, the $\lambda$ parameter appears in the action through the background charge coupling 
\ie
 {1-2\lambda\over 4\pi}\int d^2 z \sqrt{g} \phi R\,.
\fe
Here $\phi$ is a real scalar from the bosonization of the $bc$ ghosts via $J^{\rm gh}_z=\pa_z\phi,~J^{\rm gh}_{\bar z}=\pa_{\bar z}\phi$. 
}
\ie
h_b=\bar h_{\bar b}=\lambda,\quad h_c=\bar h_{\bar c}=1-\lambda\,,
\fe
and the conformal central charges
\ie
c_L=c_R=1-3(2\lambda-1)^2 \,.
\fe
The $bc$ CFT contains a ghost number current
\ie
(J^{\rm gh}_z,J^{\rm gh}_{\bar z})=(-bc,-\bar b\bar c)\,,
\fe
and its dual which generate vector and axial $U(1)$ ghost number symmetries. The $b,c$ ghosts have charges $\mp 1$ respectively with respect to $J^{\rm gh}_z$, while the $\bar b,\bar c$ ghosts have charges $\mp 1$ respectively with respect to $J^{\rm gh}_{\bar z}$. 
The vector ghost number symmetry has a gauge-gravity anomaly
\ie
\nabla^\m J^{\rm gh}_\m= {1-2\lambda \over 2} R\,,
\label{bca}
\fe
where $R$ is the Ricci curvature scalar. This is an \textit{exotic} anomaly that solves the WZ consistency condition but does not arise from the usual descent procedure \cite{Nakayama:2018dig,Chang:2020aww}.\footnote{On general grounds, they are not admissible in a unitary theory with a normalizable vacuum \cite{Chang:2020aww}. 
}
Yet the theory has a well defined boundary condition $\cB$ at ${\rm Im}(z)=0$ (i.e. $ x_\perp=0$) given by 
\ie
c(z)=\bar c(\bar z),\quad b(z)=\bar b(\bar z)\,,
\label{ghostbc}
\fe
that preserves the vector ghost number symmetry and is essential for formulating worldsheet string theory on Riemann surfaces with boundaries. 

At the special value $\lambda=1/2$, the $bc$ CFT is identical to a free Dirac fermion (and the theory becomes unitary) and $J^\m_{\rm gh}$ is nothing but the fermion number current. Indeed the anomaly \eqref{bca} vanishes in this case in accordance with Theorem~\ref{thm:vanpa}. 

One may wonder where the CFT arguments in the last section fails for the general $bc$ CFT with the symmetric boundary \eqref{ghostbc} preserving the anomalous ghost number symmetry, since the vanishing conditions \eqref{Jv} and \eqref{Tv} are obviously satisfied by the ghost number current and $bc$ stress-tensor in the boundary limit \ie
\lim_{x_\perp\to 0} J^\perp_{\rm gh}(x)=\lim_{x_\perp\to 0} T^{\sigma \perp}(x)=0\,,
\label{bcvancond}
\fe
away from other operator insertions. Recall earlier a tension between such vanishing conditions and the anomalous current Ward identity in the presence of a boundary was what led us to conclude that the anomaly must be zero. 
To this end, we emphasize an important feature of the anomaly \eqref{bca} (for $\lambda\neq {1\over 2}$) compared to the conventional 't Hooft anomalies is that it's even under spacetime parity. Consequently when the theory is placed on a manifold with boundary, there exists symmetric parity-even terms localized on the boundary that modifies \eqref{bca}, which is not possible in the parity-odd case.

More explicitly, the relevant current Ward identity (compared to \eqref{aJw} and \eqref{JTTw} for the parity-odd anomalies) in the absence of a boundary takes the following form
\ie
\la \pa_\m J_{\rm gh}^\m(x) T_{\A\B}(x')\ra = {1-2\lambda\over 2} (\pa_\A \pa_\B - \pa^2 \D_{\A\B})\D^2(x-x')\,.
\label{JTw}
\fe
With a symmetric boundary preserving the current $J^\m_{\rm gh}$ (e.g. as \eqref{ghostbc} in the $bc$ CFT), the Ward identity \eqref{JTw} admits the following modification localized on the boundary (which obeys the WZ consistency condition)
\ie
\la \pa_\m J_{\rm gh}^\m(x) T_{\A\B}(x')\ra_\cB \supset  {\kappa\over 2} (\D_{\A\B} 
\D_{\C \perp} - 2\D_{\C(\A}\D_{\B)\perp})\pa^\C_{(x')}\D^2(x-x')\D(x'_\perp)\,.
\label{JTbw}
\fe
It is equivalent to the the following local modification of \eqref{bca} upon coupling to background metric, 
\ie
\nabla_\m J^\m_{\rm gh}={1-2\lambda\over 2}R(x)-\kappa K(x) \D(x_\perp)\,,
\fe
where $K(x)$ is the extrinsic curvature. 

In the $bc$ CFT, the value of $\kappa$ is fixed
\ie
\kappa={2\lambda-1}\,.
\fe
For $\lambda=m+1\in \mZ$, this follows from counting zero modes of the $b,c$ ghosts with the boundary condition \eqref{ghostbc} (see e.g. \cite{Alvarez:1982zi}). For general $\lambda$, one should be able to derive this by analyzing the two-point function of $T_{\m\n}$ and $J^{\rm gh}_\m$ in the presence of the boundary, but we will not pursue it here. Consequently
\ie
\nabla_\m J^\m_{\rm gh} = {1-2\lambda\over 2}  (R(x) +2K(x)\D(x_\perp))\,,
\label{bcab}
\fe
and
\ie
\int_\cM \sqrt{g} \nabla_\m J^\m_{\rm gh} =2\pi (1-2\lambda) \chi(\cM) \,,
\fe
where $h$ denotes the induced metric on the boundary $\Sigma=\pa \cM$ and the above follows from the Gauss-Bonnet theorem.

Therefore for the $bc$ CFT, the full current Ward identity in the presence of the boundary gives
\ie
\lim_{x'_\perp \to 0}\la \pa_\m J_{\rm gh}^\m(x) T_{\sigma \perp}(x')\ra_\cB={1-2\lambda\over 2}
\pa^{(x')}_\perp \pa^{(x')}_\sigma (\D^2(x-x')\theta(x'_\perp))\,.
\label{JTbf}
\fe
Note that by a c-number redefinition of the ghost current
\ie
\tilde J_{\rm gh}^\m (x)=  J_{\rm gh}^\m(x) -{1-2\lambda\over 2}\D^\m_\perp K(x) \theta(x_\perp)\,,
\label{gcrd}
\fe
we can completely absorb the RHS of \eqref{JTbf}
\ie
\lim_{x'_\perp \to 0}\la \pa_\m \tilde J_{\rm gh}^\m(x) T_{\sigma \perp}(x')\ra_\cB=0\,,
\fe
which is indeed consistent with \eqref{bcvancond}. 

{In fact we have the stronger result
\ie
\la \pa_\m \tilde J_{\rm gh}^\m(x) T_{\n \rho}(x')\ra_\cB=0\,,
\label{JTbff}
\fe
on flat space after using the Gauss-Codazzi equation in $D=2$. Note that the anomaly \eqref{bca} and \eqref{bcab} implies that the symmetry defect for the $U(1)$ ghost number symmetry $\cL_\eta=e^{i \eta \oint  \star J_{\rm gh}}$ has an isotopy anomaly \cite{Chang:2018iay,Chang:2020aww}. The redefinition of the ghost number current \eqref{gcrd} and consequently the condition \eqref{JTbff} ensure that the modified symmetry defect $\tilde\cL_\eta=e^{i \eta \oint \star\tilde J_{\rm gh}}$ is topological in the presence of the boundary.
}

The parity-even exotic anomaly \eqref{bca} straightforwardly generalizes to higher dimensions \cite{Chang:2020aww},
\ie
d\star J= k e(R)\,,
\label{dJa}
\fe
where $e(R)$ denotes the Euler class in $D=2n$ dimensions and its integral over a closed manifold  $\cM$
\ie
\chi(\cM)=\int_\cM e(R)
\label{}
\fe
computes the Euler characteristic of $\cM$. Such anomalies naturally arise in supersymmetric QFTs upon (partial) topological twist \cite{Witten:1988ze}.\footnote{The $bc$ CFT is related to the Dirac fermion CFT (and supersymmetric cousins) by a similar twist that involves shifting the stress tensors by the (anti)holomorphic derivatives of the ghost number currents.} Perhaps the most well-studied examples come from the Donaldson-Witten twist of $\cN=2$ supersymmetric QFTs in $D=4$ \cite{Witten:1994cg,Witten:1995gf,Moore:1997pc,Losev:1997tp}. The physical theory before twisting has $U(1)_R\times SU(2)_R$ R-symmetry. Here $J$ corresponds to the $U(1)_R$ current of the supersymmetric theory, and $k$ is proportional to the mixed $U(1)_R$\,-\,$SU(2)_R$ anomaly. After twisting (which identifies the $SU(2)$ components of the spin connection with the $SU(2)_R$  background gauge field), $k$ corresponds to a mixed $U(1)_R$\,-gravity anomaly as we have also seen in the $bc$ CFT.\footnote{To be more precise, the full Ward identity for the $U(1)_R$ current in the Donaldson-Witten theories take the following form \cite{Shapere:2008zf}
\ie
d\star J= k e(R)+ k' p_1(T)\,,
\label{DWdJ}
\fe
where $k'$ is positive in physical CFTs (before twisting). This is because $k'$ is proportional to the coefficient of the unique conformal structure (conformal $c$-anomaly) in the stress-tensor two-point function which is reflection positive.
} If the $D=4$ theory is conformal, $k$ is proportional to a combination $2a-c$ of the conformal anomalies $a$ and $c$ \cite{Shapere:2008zf}.

In the case with boundary, similar to the $bc$ CFT, we can imagine a modified (integrated) current Ward identity of the following form,
\ie
d\star J= k (e(R)+\Phi\D(x_\perp)dx_\perp)\,,
\fe
where $\Phi$ is a $2n-1$-form that participates in the Gauss-Bonnet-Chern theorem for manifold $\cM$ with boundary $\Sigma$
 \cite{Chern,GBC}
 \ie
 \chi(\cM)=\int_\cM e(R) +\int_\Sigma \Phi  \,.
 \fe
As before, upon a redefinition of the current by
\ie
\tilde J =J-\star \Phi \theta(x_\perp)\,,
\fe
we have topological symmetry defects from
$e^{i\eta \oint \star \tilde J}$ with the putative symmetric boundary. 
It would be interesting to see if such symmetric boundary conditions arise in anomalous non-unitary QFTs (either from topological twist or not).\footnote{Note that the presence of a parity-odd anomaly in \eqref{DWdJ} suggests that such a symmetric boundary condition is likely not possible for the Donaldson-Witten theories.} 

Finally we emphasize that, in principle our Theorem~\ref{thm:vanpa} can fail in more dramatic ways for non-unitary theories (since unitarity was explicitly used in the proof). However our knowledge of non-unitary QFTs is rather limited and the familiar examples are often non-unitary only in mild ways.
It would be interesting to explore more systematically non-unitary QFTs including their anomalous symmetries and boundary conditions.\footnote{See \cite{Gorbenko:2018ncu,Gorbenko:2018dtm,Giombi:2019upv,Jepsen:2020czw} for recent studies of complex CFTs, which are non-unitary is a more dramatic fashion than the non-unitarity Virasoro minimal models and the $bc$ CFT in $D=2$.}

\subsection{Implications for Unitary CFTs}
\label{sec:perteg}
 Our general results in the previous sections give rise to strong constraints on possible boundary conditions of a given CFT $\cT$, without relying on any Lagrangian descriptions. Here we discuss a few examples for illustration.  
 
First of all, a theory $\cT$ with pure gravitational anomalies cannot admit Lorentz invariant boundary conditions (see Corollary~\ref{cor:pertg} and discussions therein). This includes familiar 2d CFTs with non-vanishing $c_L-c_R$ (e.g. chiral bosons and fermions),
 \ie
 \cI_4 [\cT]\supset {c_L-c_R\over 24} p_1(T) \,,
 \fe
 but also the mysterious strongly-coupled 6d $\cN=(2,0)$ superconformal field theory (SCFT) labelled by an ADE Lie algebra $\mf{g}$, whose anomaly polynomial takes the form \cite{AlvarezGaume:1983ig,Harvey:1998bx,Intriligator:2000eq,Yi:2001bz}
\ie
 \cI_8 [\cT]\supset  {1\over 4!} {r_{\mf{g}}\over 8} 
 (  p_1(T)^2-4p_2(T) )\,,
 \fe
 where $r_{\mf{g}}$ denotes the rank of $\mf{g}$.
 
 While there is no pure perturbative gravitational anomalies for $D=4$ theories, there can be mixed $U(1)$-gravitational anomalies. This happens for a large class of $D=4$ CFTs with $\cN=1$ supersymmetry (e.g. the super-QCD in the conformal window) whose anomaly polynomial takes the form
 \ie
 \cI_6 [\cT]\supset{1\over 3!} \left( {k_{R}} c_1(F_{U(1)_R}) p_1(T) + k_{RRR} c_1(F_{U(1)_R})^3
 \right)
 \fe
 where the $U(1)_R$ denotes the R-symmetry which is a part of the $\cN=1$ superconformal symmetry. From Theorem~\ref{thm:vanpa}, we deduce that an $\cN=1$ SCFT $\cT$ may have a $U(1)_R$ preserving boundary condition only if $k_R=k_{RRR}=0$, which is not possible unless the SCFT contains no local degrees of freedom.\footnote{Here we have used the relation between the conformal central charges and the 't Hooft anomalies \cite{Anselmi:1997am,Anselmi:1997ys} and the bounds on the conformal central charges in unitary CFTs \cite{Hofman:2008ar}.} The $\cN=4$ super-Yang-Mills (SYM) with gauge algebra $\mf{g}$ is a particular $\cN=1$ SCFT with enhanced superconformal symmetry and an anomaly polynomial free from gauge-gravity anomalies, 
 \ie
 \cI_6[\cT]= {d_\mf{g}\over 2}c_3(F_{SU(4)_R})\,,
 \label{n4anom}
 \fe
 where $SU(4)_R$ is the enhanced R-symmetry and $d_\mf{g}$ is the dimension of $\mf{g}$. The $\cN=4$ SYM is known to admit a large family of half-BPS boundary conditions \cite{Gaiotto:2008sa,Gaiotto:2008ak} that preserve half of the supersymmetries and the R-symmetry maximal subgroup $SU(2)_H\times SU(2)_C \subset SU(4)_R$. It is easy to check that the anomaly \eqref{n4anom} indeed trivializes in this subgroup, in agreement with our general results.\footnote{Moreover, inside another maximal subgroup  $SO(2)_R\times SU(2)_R \times SU(2)_F \subset SU(4)_R$, the $SU(2)_F$ subgroup has a global Witten anomaly \cite{Witten:1982fp} (see also \cite{Tachikawa:2018rgw} for such an anomaly in general strong coupled $D=4$ CFTs) which in modern perspective is captured by the nontrivial element of the bordism group
 $\Omega_5^{Spin}(BSU(2))=\mZ_2$. Looking ahead, as we will argue in Section~\ref{sec:genanom}, such an global anomaly also obstructs a symmetric boundary condition preserving this $SU(2)_F$. Indeed the $SU(2)_F$ symmetry is broken by the known boundary conditions \cite{Gaiotto:2008sa,Gaiotto:2008ak}. This anomaly occurs at height $k = 1$ in the notation of Section \ref{sec:genanom}.
 }

	\section{Symmetry Defects and Group Cohomology}
	   \label{sec:defect}
	  Now and in the remainder of the paper we will generalize to the case where the global symmetry $G$  is not necessarily continuous, and thus a Noether current may not be available. We will still need some notion of locality for the symmetry action, and so we will associate the symmetry generator with special defect operators and consider correlation functions defined in the presence of networks of these operators.
	In this section, we will suppose these defect operators are topological, which precludes so called gauge-gravity anomalies we will discuss in Section \ref{sec:gganom} below.

	\subsection{$G$-foams and Background Gauge Fields}
	\label{sec:gfax}
	
   We want to say that a theory $\cT$ has an unbroken $G$-symmetry if we can define correlation functions in the presence of a network of $G$-symmetry defects. The specific kind of network we want is what we will call a  \textbf{$G$-foam}, which consists of a collection of co-oriented embedded closed hypersurfaces labelled by elements of $G$, meeting transversely along their boundaries. This means that in any small open neighborhood, the foam is Poincar\'e dual to a triangulation. Any collection of embedded closed hypersurfaces can be infinitesimally perturbed to satisfy this property. Furthermore, we will require the labels to satisfy the following axiom.
    \begin{ax} \textbf{(Flatness)}
    At a codimension-two junction (where three hypersurfaces of the $G$-foam meet) the path-ordered product of the $G$ labels along an oriented loop linking the junction is the identity, where a hypersurface with label $g$ contributes $g$ to this product if its co-orientation agrees with the orientation of the loop and $g^{-1}$ otherwise.
    \label{ax:0}
    \end{ax}
    
    Correlation functions are allowed to depend on local operator insertions, the metric and tangent structure of spacetime, etc., but the dependence on the $G$-foam is constrained to obey the following axioms.

    \begin{ax}    Correlation functions are invariant under passing a local operator from one side of a $g$-hypersurface to the other while transforming the local operator by the group element $g$ or $g^{-1}$ if it goes with or against the co-orientation, respectively.
 \label{ax:1}
    \end{ax}
    \begin{ax} 
    Correlation functions are invariant under introducing or removing spherical components of the foam whose interior doesn't contain any foam or other operator insertions.
     \label{ax:2}
    \end{ax}
    \begin{ax} 
    Correlation functions are invariant under smooth isotopies of the foam. (This axiom is relaxed in cases with gauge-gravity anomaly in Section \ref{sec:gganom}.)
     \label{ax:3}
    \end{ax}
    \begin{ax} 
    We can perform a recombination of the $G$-foam, changing its local topology and only affecting the correlation function by a phase.
      \label{ax:4}
    \end{ax}
    
    We will consider two realizations of $G$-symmetries of the same theory $\cT$ (meaning two definitions of correlation functions satisfying the above axioms) equivalent if they are related by redefining these correlation functions by phase factors associated with the point-like singularities of the foam. We will see this ambiguity corresponds to adding a local gauge-invariant counterterm to the action.
    We call a $G$-symmetry \textit{anomalous} if it is not equivalent to one that assigns trivial phase factors to recombinations of the $G$-foam in Axiom~\ref{ax:4}.

    One way to define correlation functions in the background of a $G$-foam is to divide spacetime up into the open regions cut out by the foam, and then impose boundary conditions such that the limit of a field from one side of a hypersurface labelled by $g$ equals the limit of the same field transformed by $g$ coming from the other side of the hypersurface, according to the co-orientation. This is captured by Axiom~\ref{ax:1}. For a symmetry associated with a Noether current, these topological defect operators are given by integrating the component of the current normal to the hypersurface. In a Hamiltonian picture, we can create these defects by applying a symmetry generator in a fixed region \cite{Else_2014}.
    
    We can associate a $G$ gauge field $A$ to such a foam by defining the holonomy ${\rm P}e^{\oint_\gamma A}$ along a closed oriented loop $\gamma$ to be $g_1 \ldots g_n$ where $g_j$ are the labels of hypersurfaces encountered along $\gamma$ (contributing $g$ or $g^{-1}$ depending on whether $\gamma$'s orientation agrees or disagrees with the co-orientation of the corresponding hypersurface) in the order that they are encountered. The holonomy ${\rm P}e^{\oint_\gamma A}$ only depends on the choice of the starting point up to a conjugation in $G$. These holonomies determine $A$ up to gauge transformations. We can thus think of the $G$-foam as a kind of Poincar\'e dual of a $G$ gauge field.
    
    The flatness Axiom~\ref{ax:0} corresponds to the condition that the curvature of this gauge field vanishes, or in other words that $\oint_\gamma A = 0$ around any contractible loop. This forbids the scenario with a $g$-hypersurface simply ending somewhere. Such an object is like a magnetic flux, and while it may be possible to define correlation functions in the presence of these, for our purposes we will not need to. In fact, for compact $G$, the anomaly is completely determined by its finite subgroups, and hence it suffices to study only flat connections. This follows from the result proven in \cite{mathoverflow}. This ensures that the $G$-foam approach is sufficient to characterize anomalies of all compact Lie groups.\footnote{In particular, as we will see, this gives another proof for the obstruction theorems in Section~\ref{sec:conformal} for symmetric boundaries in the presence of perturbative anomalies.
    }
    
    In terms of the Poincar\'e dual gauge field $A$, Axioms~\ref{ax:2},~\ref{ax:3}, and~\ref{ax:4} capture the gauge invariance of the correlation functions. Axiom~\ref{ax:2} corresponds to applying a local symmetry transformation in some region away from any operator insertions. Axiom~\ref{ax:3} corresponds to applying a local symmetry near another defect, causing it to move. Below, we will see that this axiom is violated in systems with gauge-gravity anomalies. For symmetries with a Noether current $j^\mu$, these conditions are equivalent to $\partial_\mu j^\mu = 0$ in any region not containing a defect.

    \subsection{Group Cohomology Anomalies and Wess-Zumino Consistency}
    
    Axiom~\ref{ax:4} encodes the anomaly of the relevant symmetries. To understand it, we first make the following observation.

    \begin{lemma} (Wess-Zumino Consistency)
    Any sequence of isotopies and recombinations occurring in a small region which takes a $G$-foam back to itself results in a trivial phase factor.
    \label{lem:WZ}
    \end{lemma}
    \begin{proof} Suppose this were not the case. Then in any correlation function, away from any of the operator insertions, we would be able to apply Axioms~\ref{ax:2}, \ref{ax:3} and \ref{ax:4} to create and then destroy a ``bubble" of the foam which undergoes the offending recombination and then disappears the same way it came to be. The result is an identity relating the original correlation function to itself with a nontrivial phase factor. Thus, the correlation function (which was arbitrary) must vanish.
    \end{proof}
    
    We will use a very similar argument to show our main result for boundaries in subsequent sections. We note that this is the discrete analog to the Wess-Zumino consistency conditions \eqref{WZ} in Section \ref{sec:conformal}. Indeed, \eqref{WZ} says that the action of infinitesimal gauge transformations integrates to an action of the group of gauge transformations. In particular, if we have a loop of infinitesimal gauge transformations that ends in the identity, it must produce a trivial phase factor in the correlation functions, or else they must vanish.
    
    It follows from Lemma \ref{lem:WZ} that we can perturb our recombinations to be generic without changing the phase factors in Axiom \ref{ax:4}. Generic recombinations  of $G$-foams are sequences of one elementary recombination, known as a Pachner move \cite{PACHNER1991129}, analogous to the $F$-move or crossing relation in 2D foams (which are networks of line defects). To see this, it is useful to think about the ``movie" of a sequence of isotopies and recombinations as a $G$-foam in $D+1$ dimensions. This foam still satisfies the flatness axiom. Perturbing our sequence of recombinations to be generic makes this $G$-foam locally Poincar\'e dual to a triangulation. In particular, the point-like singularities where recombinations occur are all Poincar\'e dual to a $D+1$-simplex, such that each incident hypersurface is associated with one of the $D(D+1)/2$ edges. Let us order the vertices of this $D+1$ simplex $x_0, \ldots, x_{D+1}$. Because of flatness, the $G$ labels on hypersurfaces are all determined by those hypersurfaces passing through the edges $(x_0 x_1), (x_1 x_2), \ldots, (x_D, x_{D+1})$ since every vertex-to-vertex path on the boundary of the simplex can be deformed to one that runs only along these edges. We can thus write our phase factor in Axiom~\ref{ax:4} as $e^{2\pi i \omega}$ with $\omega$ a function
    \[\omega:G^{\times D+1} \to U(1) = \mathbb{R}/\mathbb{Z}\,,\]
    evaluated on these particular labels. We return to this point of view when we derive anomaly in-flow in Section \ref{sec:ainf}. The elementary recombination for $D=2$ and its ``movie" are shown in Fig.~\ref{fig:singularity}.

    \begin{figure}[!htb]
        \centering
        \includegraphics[width=10cm]{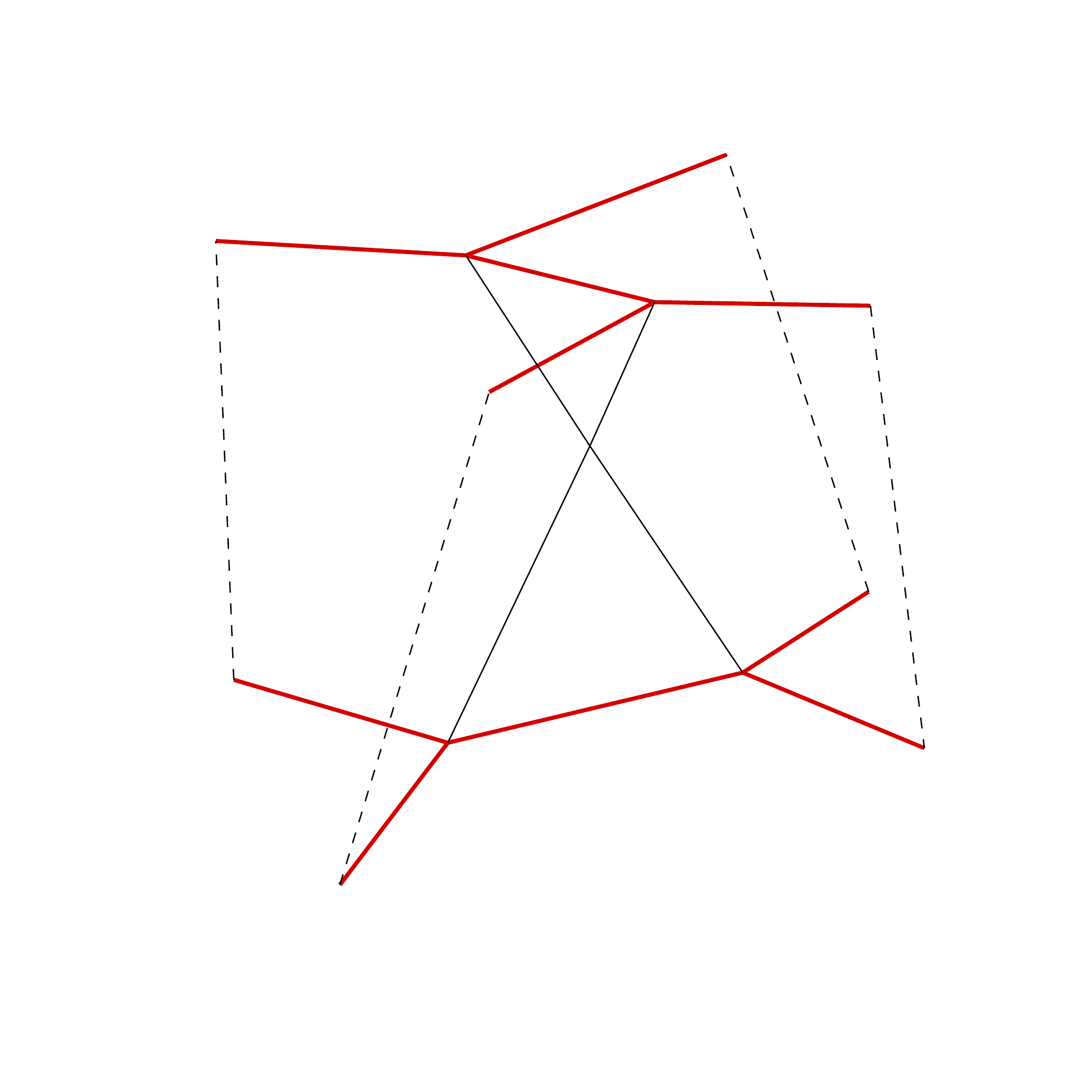}
        \caption{The elementary recombination for $D=2$ (sometimes called an $F$-move or crossing relation), drawn here as relating the red foam at bottom to the red foam at top, is captured by a point-like singularity of a three dimensional foam which is Poincar\'e dual to a tetrahedron and for which the initial and final foams form the top and bottom boundaries.}
        \label{fig:singularity}
    \end{figure}
    
    Lemma~\ref{lem:WZ} implies the group cocycle equation for $\omega$ \cite{brown2012cohomology}:\footnote{This cocycle condition is slightly modified when certain symmetries of $G$ are spacetime-orientation reversing. We revisit this point below.}
    \[
    \begin{gathered}
    (\delta \omega)(g_1,\ldots,g_{D+2}) := \\
    \omega(g_2,\ldots,g_{D+2}) + \left(\sum_{j=1}^{D+1} (-1)^j \omega(g_1, \ldots, g_{j-1} g_j, g_{j+2},\ldots,g_{D+2}) \right) + (-1)^{D+2} \omega(g_1,\ldots, g_{D+1}) = 0,    
    \end{gathered}
    \]
    modulo $1$, since it appears in $e^{2\pi i\omega}$. MacLane's coherence theorem (see Chapter 7 of \cite{lane1998categories} as well as \cite{KAPRANOV1993119}) tells us that this is in fact the \emph{only} condition contained in Lemma~\ref{lem:WZ}. This is quite analogous to the fact that there is a single elementary recombination we need to consider in Axiom~\ref{ax:4}. Indeed, this equation is associated to the single elementary ``recombination of recombinations". This point of view is key in anomaly in-flow, which we return to in Section \ref{sec:ainf}.
    
    We now consider redefining our correlation functions in a background $G$-foam by assigning phase factors to point-like junctions.\footnote{It is possible to locally assign phases to higher dimensional defects in a correlation function by invoking the embedded geometry or tangent structure of the defect, slightly violating our axioms. We do exactly this in Section \ref{sec:gganom} and find there are more ambiguities for $\omega$ than are described in this section.} The point-like junctions are Poincar\'e dual to a $D$-simplex and as above the $G$ labels of the incident hypersurfaces are determined by a $D$-tuple of $G$ elements. A choice of phase factor for each point-like junction is thus a function $\alpha: G^{\times D} \to U(1)$. The effect of this redefinition is to shift $\omega \mapsto \omega + \delta \alpha$, where once again the group coboundary operator $\delta$ appears. This redefinition of the correlation functions is equivalent to adding a counterterm $i \int \alpha(A)$ to the action. Thus, while $\omega$ is somewhat arbitrary, its cohomology class $[\omega] \in H^{D+1}(BG,U(1))$ is invariant under these redefinitions.
    
    If we think about our $G$-foam as a gauge field, and recombinations as gauge transformations, $[\omega] \neq 0$ means we cannot find a counterterm of the form $e^{i \int \alpha(A)}$ which allows us to gauge the theory. Thus $[\omega]$ captures the 't Hooft anomaly. We will refer to this type of anomaly as a \textit{group cohomology anomaly} to contrast with a slight generalization of the above which we will discuss later. For an alternative description of the group cohomology anomaly in a similar spirit but in the Hamiltonian picture, see \cite{Else_2014}. By anomaly in-flow (see Section \ref{sec:ainf}) these anomalies are related to group cohomology SPT phases \cite{pollmannbergturneroshikawa,Wen_2013,CGLW}.
    
    For finite $G$, there is an isomorphism $H^{D+1}(BG,U(1)) = H^{D+2}(BG,\mathbb{Z})$ given by $\omega \mapsto \delta\omega$ (recall $\delta \omega \in \mathbb{Z}$). The anomaly polynomial $\cI_{D+2}[\cT]$ defined in Section \ref{sec:conformal}, for a pure gauge anomaly, defines a class in $H^{D+2}(BG_0,\mathbb{Z})$ for a connected Lie group group $G_0$, and the group cohomology anomaly for each of its finite subgroups $G \subset G_0$ is obtained by restriction and then inverting the isomorphism above. The collection of these restrictions actually determines $\cI_{D+2}[\cT]$ exactly \cite{mathoverflow}. The constructions of this section allow us to define $\omega$ also for continuous symmetry groups directly, subject to the condition that $\omega$ and the counterterms $\alpha$ are measurable functions on $G$. This extra condition ensures that we can integrate our extended correlation functions over the gauge group. This defines the so-called Borel measurable group cohomology $\cH$, which for any compact Lie group also satisfies $\cH^{D+1}(G,U(1)) = H^{D+2}(BG,\mathbb{Z})$ \cite{stasheff1978,wagemann2013cocycle}, capturing both the anomalies of finite groups as above, but also the anomaly polynomials of the connected parts.
    
    \subsection{Spacetime-Orientation-Reversing Symmetries}
    
    Above we tacitly assumed that the symmetries of $G$ were internal, unitary symmetries. More generally we can consider spacetime-orientation-reversing (SOR) symmetries, such as time reversal or reflection symmetries. An advantage of the spacetime picture is that it is more or less clear how to generalize $G$-foams and gauge fields to this setting \cite{kapustin2014symmetry,Witten_2016}.
    
    Essentially to define correlation functions we must choose an orientation inside each open region cut out by the $G$-foam and define the theory inside with respect to that orientation, such that across a hypersurface labelled by an SOR symmetry, the local orientation flips.
    
    This local orientation has two main effects. First, the global structure of the foam is constrained by the topology of spacetime, such that a curve along with the orientation flips an odd number of times must be an orientation-reversing cycle of spacetime. Otherwise, Axioms \ref{ax:0}, \ref{ax:1}, \ref{ax:2}, and \ref{ax:3} are unmodified.
    
    Second, the phase factor $\omega$ from Axiom \ref{ax:4} can depend on this local orientation, but must satisfy the property that if we reverse all the local orientations surrounding the junction where the recombination occurs, keeping everything else the same, then $\omega \mapsto -\omega$ (this holds for the counterterms as well). This is a consequence of Lemma \ref{lem:WZ}. In fact, redoing the argument for consistency of recombinations, we find that for bosonic systems $\omega$ defines a class in the \textit{twisted} group cohomology $H^{D+1}(BG,U(1)^T)$, where $U(1)^T$ indicates that SOR elements of $G$ act by conjugation on $U(1)$, which matches the group cohomology classification of SPTs \cite{CGLW}. With this generalization understood, and all coefficient groups appropriately twisted, we will not explicitly refer to SOR symmetries again.
    
    \subsection{Boundaries for Group Cohomology Anomalies}
	   \label{sec:groupcoh}
    Now let us consider the case with a boundary. We will say the symmetry $G_\cB \subset G$ is unbroken at the boundary $\cB$ if correlation functions in the theory $\cT$ can be defined in the presence of a $G_\cB$-foam which terminates along the boundary, satisfying the five Axioms in Section~\ref{sec:gfax}. In addition to the elementary recombination which can occur among junctions in the bulk, there is a second elementary recombination, which for the boundary foam looks like the bulk elementary recombination in one lower dimension, and for the bulk foam looks like a point-like junction being absorbed or emitted by the boundary. See Fig.~\ref{fig:absorb}.

    \begin{figure}[!htb]
        \centering
        \includegraphics[width=14cm]{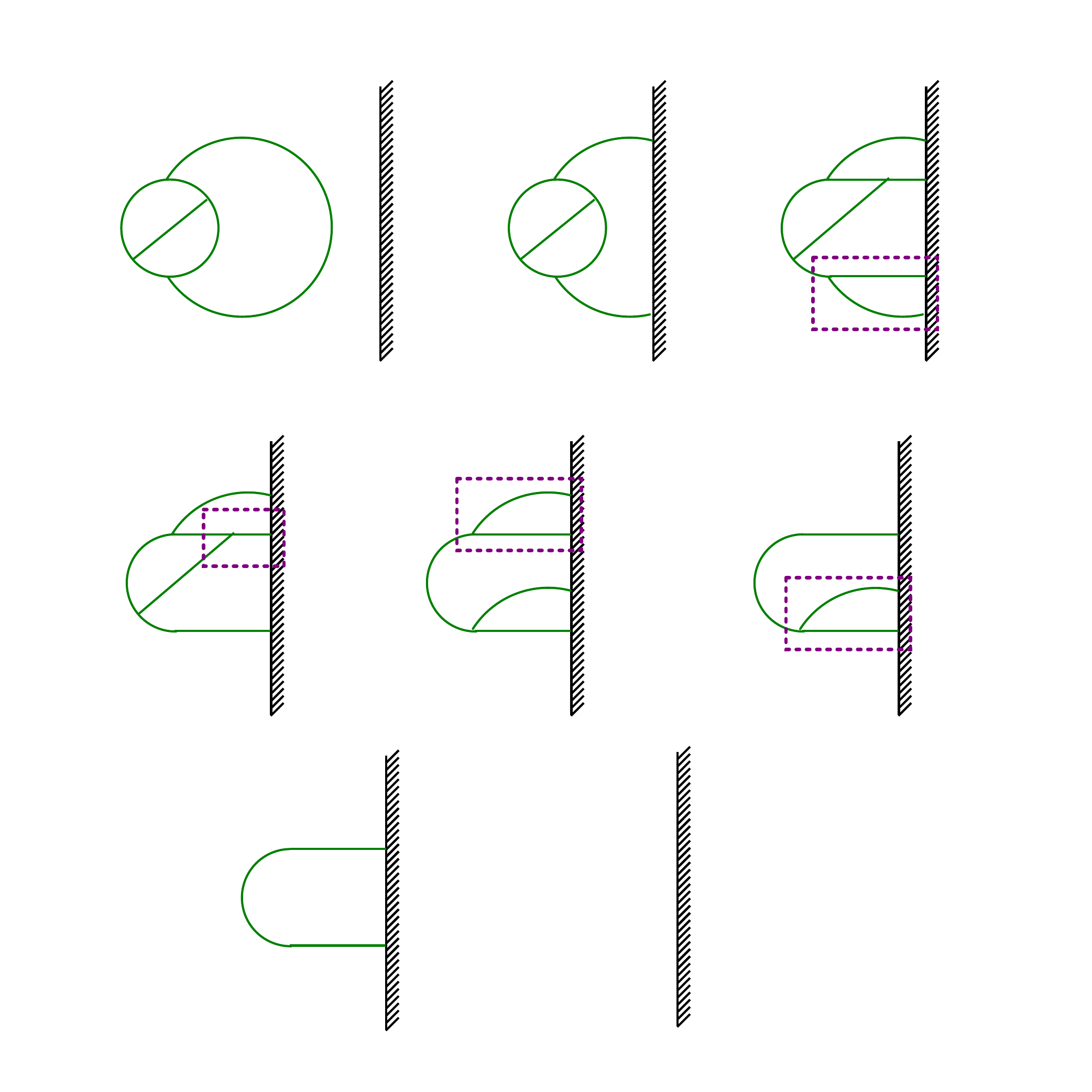}
        \caption{The argument for Theorem~\ref{propwog} illustrated for $D = 2$. At top left we have the special bubble introduced near the boundary. From left to right, top to bottom, this bubble is absorbed by the boundary. In purple we have indicated where absorbing a point-like junction causes a recombination of the boundary defects, and produces compensating phases. The condition that the boundary correlation function is non-vanishing requires that these phases precisely cancel $e^{2\pi i\omega}$ of the original bubble. This requires $\omega$ be exact in group cohomology, hence that there is no anomaly.}
        \label{fig:absorb}
    \end{figure}

    \begin{thm}
    If a theory has a boundary condition with unbroken $G_\cB$ symmetry, as we have defined it, then furthermore the 't Hooft anomaly restricted to $G_\cB$ is also trivial, i.e. $j^*[\omega] = 0 \in H^{D+1}(BG_\cB,U(1)^T)$, where $j:G_\cB \hookrightarrow G$ is the inclusion map for the subgroup $G_\cB$ and $[\omega] \in H^{D+1}(BG,U(1)^T)$ is the anomaly class.
        \label{propwog}
    \end{thm}
    
    \begin{proof}
    First we observe that there is a compactly supported ``bubble" of $G_\cB$-foam associated to any $D+1$-tuple $(g_0,\ldots,g_D) \in G_\cB^{\times D+1}$ which has the key property that it can be created or destroyed inside any correlation functions while changing its phase by $e^{\pm 2\pi i \omega(j(g_0),\ldots,j(g_D))} = e^{\pm 2\pi  i j^*\omega(g_0,\ldots,g_D)}$. This bubble is obtained by taking the elementary recombination singularity (as in Fig.~\ref{fig:singularity}), intersecting with a small $S^D$ encircling the singular point, and then stereo-graphically projecting it onto a $D$-ball. The key property follows from the fact that by a single recombination we can reduce the bubble to one which occurs as the boundary of a $G_\cB$-foam on $B^{D+1}$ with no singularity. Such a $G_\cB$-foam can be created or destroyed without introducing any phase factors.

    We can then introduce this bubble into any boundary correlation function and then push the bubble through the boundary until its gone, recovering the original correlation function. Either the correlation function (which is arbitrary) must vanish or the phase factors accrued while pushing the bubble through the boundary must cancel $\omega$. These extra phase factors appear when point-like junctions of the bulk $G_\cB$-foam are pushed through the boundary (see Fig.~\ref{fig:absorb}). As we've said these lead to recombinations of the boundary $G_\cB$-foam. Each of these junctions is associated with a $D$-tuple of $G_\cB$ elements, and if we write the phase factors as $\alpha:G_\cB^{\times D} \to U(1)$, the condition that these phase factors cancel $e^{2\pi ij^*\omega}$ is $j^*\omega = \delta\alpha$, where $\delta$ is the (twisted) group coboundary operation. In other words, $j^*[\omega] = 0 \in H^{D+1}(BG_\cB,U(1)^T)$.
    \end{proof}

    \section{Gauge-Gravity Anomalies}
       \label{sec:gganom}
    It is known that the anomalies described above are not the most general kind. In fermionic systems for instance, there could be fermions bound to the 1-dimensional junctions in the foam that modify the cocycle equation for $\omega$ (e.g. into the Gu-Wen equation \cite{supercohomology}), meaning that we must go beyond group cohomology to describe them. There could even be more dramatic effects like chiral modes bound to the defects (such as at the boundary of a $\nu = 1$ topological superconductor in 3+1D) which leads to a violation of the isotopy condition in Axiom~\ref{ax:3}. Most generally, we have to assume that the defects themselves host modes with a gravitational anomaly. For this reason, these more general anomalies are sometimes called gauge-gravity anomalies. In this section, we generalize the symmetry defect picture from the previous section to this setting and prove again that a symmetric boundary condition is only possible if the anomaly is trivial.
    
    \subsection{Generalized $G$-foams}
    
    To account for the gauge-gravity anomalies we must generalize our definition of a $G$-foam. Axioms~\ref{ax:0}, \ref{ax:1}, and \ref{ax:2} still hold, but Axiom~\ref{ax:3} must be modified to account for the possibility that the junctions themselves can host a gravitational anomaly. This has been appreciated in the SPT literature through the discovery of ``beyond cohomology" SPT phases \cite{Vishwanath_2013,Wang_2013,Burnell_2014,supercohomology,Fidkowski_2020}.
   \begin{thmbis}{ax:3}
 Each $k-1$-dimensional defect or junction is associated with an element of the group $\Omega^k$ of $k$-dimensional invertible phases,\footnote{These phases have no specified symmetry.} such that under isotopies, the correlation functions transform according to the boundary gravitational anomalies of these defect or junction labels.
      \label{ax:3p}
\end{thmbis}

\begin{thmbis}{ax:4} Upon recombination, we obtain a phase factor which is the product of the isotopy contribution above and a contribution that depends only on the combinatorics of the recombination.
      \label{ax:4p}
\end{thmbis}
    
    Axiom~\ref{ax:3p} refers both to phase factors accrued by small isotopies of defects (when they carry a perturbative gravitational anomaly) as well as large isotopies, such as taking a defect around a nontrivial cycle of the spacetime, for which there may be phase factors associated with global gravitational anomalies. For instance, a fermionic operator insertion at a point-like defect is characteristic of certain gauge-gravity 't Hooft anomalies \cite{supercohomology}. Taking the point-like defect around a loop of the spacetime with a non-bounding spin structure produces a minus sign.
    
    As before we will consider two realizations of the $G$-symmetry in a given theory $\cT$ to be equivalent if they are related by a redefinition of the correlation functions by phase factors at point-like singularities. Because we allow our correlation function to depend on the details of the embedding of the foam, we will also allow redefinition of them by decorating higher dimensional defects and junctions with invertible phases.
    
    \subsection{Beyond Cohomology Anomalies and Wess-Zumino Consistency}

    The Wess-Zumino consistency conditions of Lemma~\ref{lem:WZ} still apply, although the conditions they impose on the isotopy phases in Axiom~\ref{ax:3p} and the recombination phases in Axiom~\ref{ax:4p} are not completely understood. It is conjectured however that all of the data encodes a certain generalized cohomology invariant for $D+1$ manifolds equipped with a $G$ gauge field through a mathematical construction known as the Atiyah-Hirzebruch spectral sequence (AHSS) \cite{kitaevipamtalk,Kapustin_2017,Gaiotto_2019,shiozaki2018generalized,Garc_a_Etxebarria_2019}.\footnote{In particular the whole anomaly is expected to be expressed as an element of a certain cobordism group $\Omega^{D+1}(BG)$ which encodes the partition function of a $D+1$-dimensional SPT phase, but actually which generalized cohomology appears is not important for our arguments.}
    
    Let us briefly discuss the structure of the AHSS which will be relevant for us. We can consider the labels on $k-1$-dimensional defects as a function $\alpha_k:G^{D+1-k} \to \Omega^k$. With the convention $\Omega^0 = U(1)$, the full anomaly is specified by the functions $\alpha_0,\ldots,\alpha_D$.\footnote{If there were a gravitational anomaly, then additionally we would have an additional constant function $\alpha_{D+1} \in \Omega^{D+1}$ which encodes it.}

    The redefinition of the correlation function by phase factors and invertible phases is described by a collection of maps $\beta_k: G^{D-k} \to \Omega^k$, $k = 0, \ldots, D$. The $\alpha_k$ are modified according to
    \[\alpha_k \mapsto \alpha_k + \delta \beta_k + \Delta^D_k(\beta_{>k})\,,\]
    where the notation indicates that $\Delta^D_k$ is a function depending only on $\beta_j$ for $j > k$. The physical origin of this function is that invertible phases may also host gravitational anomalies along higher codimension defects, not just at their boundary. See Fig.~\ref{figvortex} for an example.

    \begin{figure}
        \centering
        \includegraphics[width=8cm]{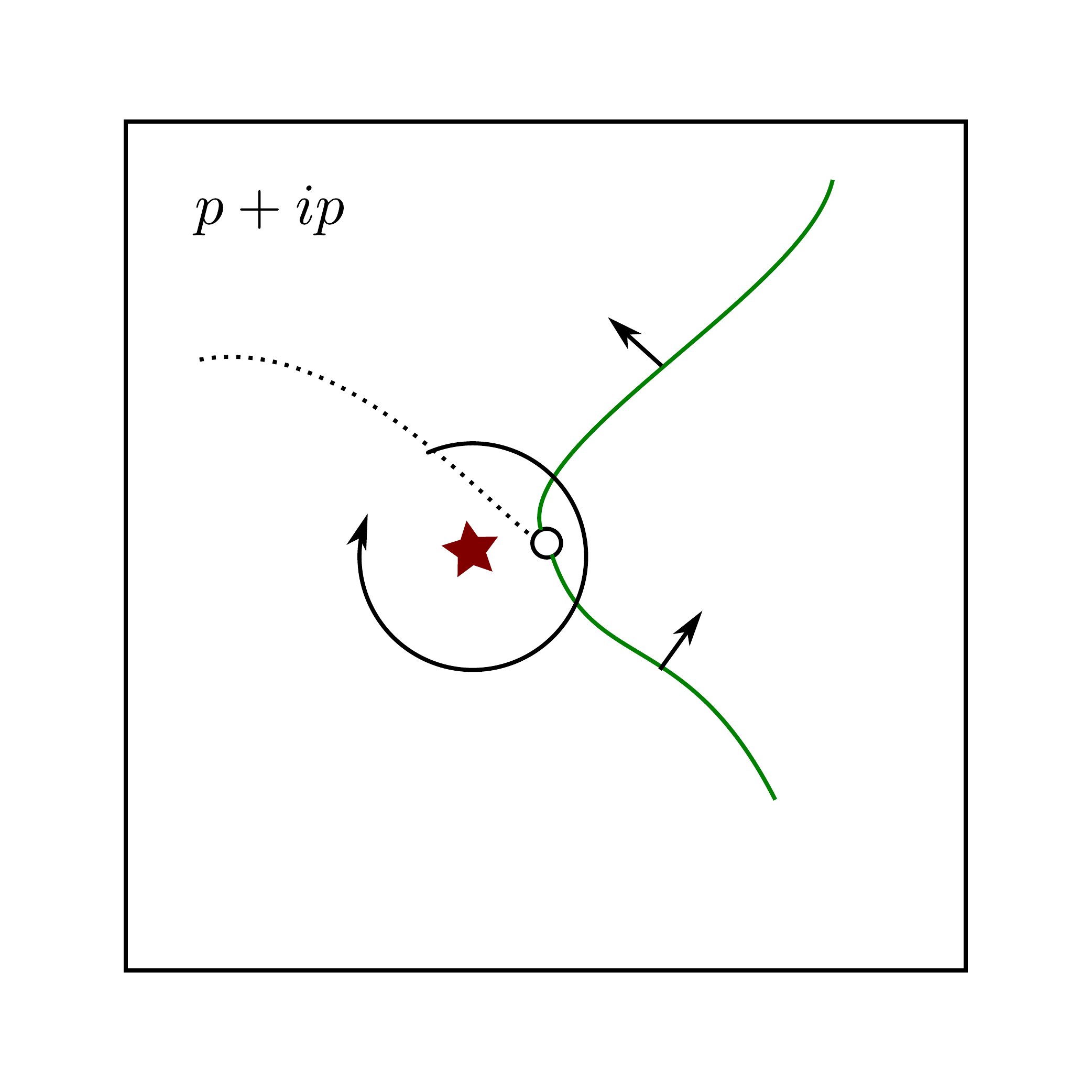}
        \caption{Here we have a spatial picture of a 2+1D fermionic system with a unitary symmetry $C^2 = (-1)^F$. At the white circle, two co-oriented $C$ defects (green) fuse to a fermion parity defect (dashed). This junction may trap a Majorana fermion. However, if it does, then by layering a $p+ip$ superconductor (so taking $\beta_3 = 1 \in \Omega^3_{Spin}$), since the fermion parity defect ends at the fusion junction, the $p+ip$ superconductor sees a vortex there (red star), which also traps a Majorana fermion \cite{Alicea_2012}. This Majorana may be paired with the other one to create a featureless fusion junction. Thus, this decoration does not actually contribute an anomaly in this symmetry class. With more work, one can show that for this symmetry class and dimension, there are in fact no nontrivial anomalies  (see Appendix C.4 of \cite{Garc_a_Etxebarria_2019}).}
        \label{figvortex}
    \end{figure}

    By the AHSS, the complete conditions of Axiom~\ref{ax:4p} can be stated as
    \[\label{generalizedcocycle}\delta \alpha_k + \Delta^{D+1}_{k+1}(\alpha_{> k}) = 0\,,\]
    where the functions $\Delta^{D+1}$ describe the ambiguities in dimension $D+1$. This has an interpretation in anomaly in-flow that our anomalous $G$-foam is associated with an anomaly-free $G$-foam in one higher dimension, but with nontrivial decorations by invertible phases and phase factors (i.e. the $\alpha_k$'s play the role of the $\beta_k$'s in one higher dimension). We return to this in Section \ref{sec:ainf} below. 
    
    Thus, understanding the functions $\Delta^D_k$ is the key to computing the classification of anomalies. They are encoded in the differentials of the AHSS, and satisfy many stringent naturality conditions, but for spin cobordism they are still not completely known. See \cite{thorngren2018anomalies} for some review and recent progress.
    
    Group cohomology anomalies are included in these more general anomalies by taking
    \[\alpha_0 = \omega, \qquad \alpha_k = 0 \quad \forall k > 0\,.\]
    We see that the ambiguity we have encountered previously is still there, given by $\beta_0$. However, there are extra ambiguities which we did not previously consider, given by $\beta_k$ for $k > 0$. This can cause a group cohomology class to become trivial after suitable decorations of the $G$-foam with invertible phases on the defects or junctions. That is, we may find $\beta_k$ such that
    \[0 = \alpha_0 + \delta \beta_0 + \Delta^D_0(\beta_{>0}) \qquad 0 = \delta \beta_k + \Delta^D_k(\beta_{>k}) \quad \forall k > 0\,.\]
    
    In terms of anomaly in-flow (and focusing on $\Omega^k$ given by smooth cobordisms) this amounts to a Dijkgraaf-Witten term which is trivial on all smooth manifold spacetimes. It is known that the first time this can happen is for six dimensional manifolds, i.e. for anomalies in $D = 5$ spacetime dimensions \cite{thom,sullivanonthom}. It is rather wonderful that all the conditions imposed by smoothness amount to the purely algebraic conditions above, which demystifies this fact: if a Dijkgraaf-Witten term is trivial on all smooth manifolds, then there is a decoration of the boundary by invertible phases which cancels the anomaly.
    
    \subsection{Boundaries for General Anomalies}
    \label{sec:genanom}
    
    Let us consider the largest $k$ for which $\alpha_k$ can not be eliminated by decorating the $G$-foam with invertible phases. We will call this the \textbf{height} of the anomaly. By the AHSS, in this case we can take $\alpha_j = 0$ $\forall j > k$. It follows from Lemma~\ref{lem:WZ} that $\delta \alpha_k = 0$ (cf. \eqref{generalizedcocycle}). We will interpret this as a generalized conservation law, show that it is violated at a hypothetical symmetric boundary, and argue that this is unphysical. Although the AHSS implies this conservation law, we feel that a physical formulation and proof would be quite desirable, and so we will make the conjecture
    \begin{conj}\label{conjgravbdy}
    A system with a gravitational anomaly admits no boundaries.
    \end{conj}

    \begin{thm}

    Assuming the conjecture, if a theory $\cT$ has a boundary condition which preserves a symmetry $G_\cB \subset G$, then $G_\cB$ is anomaly-free.
        \label{propwg}
    \end{thm}

    \begin{proof}
    
    It is apparent that our argument for Theorem~\ref{propwog}, which amounts to the case $k = 0$, must be modified to account for gauge-gravity anomalies at the boundary. Indeed, we used isotopy invariance of the boundary defects to conclude $\alpha_0 = j^*\omega$ is exact. In a more general foam, according to our modified axioms, the boundary defects are not necessarily isotopy invariant, but can carry modes with a gravitational anomaly. This gravitational anomaly in turn may be associated with an invertible phase in one higher dimension, which is naturally associated with the bulk defect which ends on the boundary, giving us a set of $\beta_k$'s as above. The condition that the phase factor we obtain in spawning our defect bubble and then absorbing it into the boundary is one amounts to the condition that if we decorate these bulk defects according to the $\beta_k$'s (including $\beta_0$), the anomaly is trivial.

    To see this, we can picture the boundary as an interface, where on the other side of the boundary, which we will call the ``subliminal" side, we allow a $G_\cB$-foam whose contribution to the correlation function is given by the invertible phases $-\beta_k$ defined above. See Fig.~\ref{fig:subliminal}. When we push defects from the bulk into the boundary, we can let these defects pass into the subliminal side as well as leaving defects on the boundary as before. Because of the decoration by invertible phases, the isotopy non-invariance of the boundary defects is restored for the whole foam by in-flow. If we create one of our special bubbles and push it through the boundary, the correlation function does not change, but the bubble is now entirely on the subliminal side and so its associated phase factor is completely accounted for by the decorations and boundary recombination phases accrued while pushing the bubble through the boundary, giving $\beta_0$.
    
    \begin{figure}
        \centering
        \includegraphics[width=7cm]{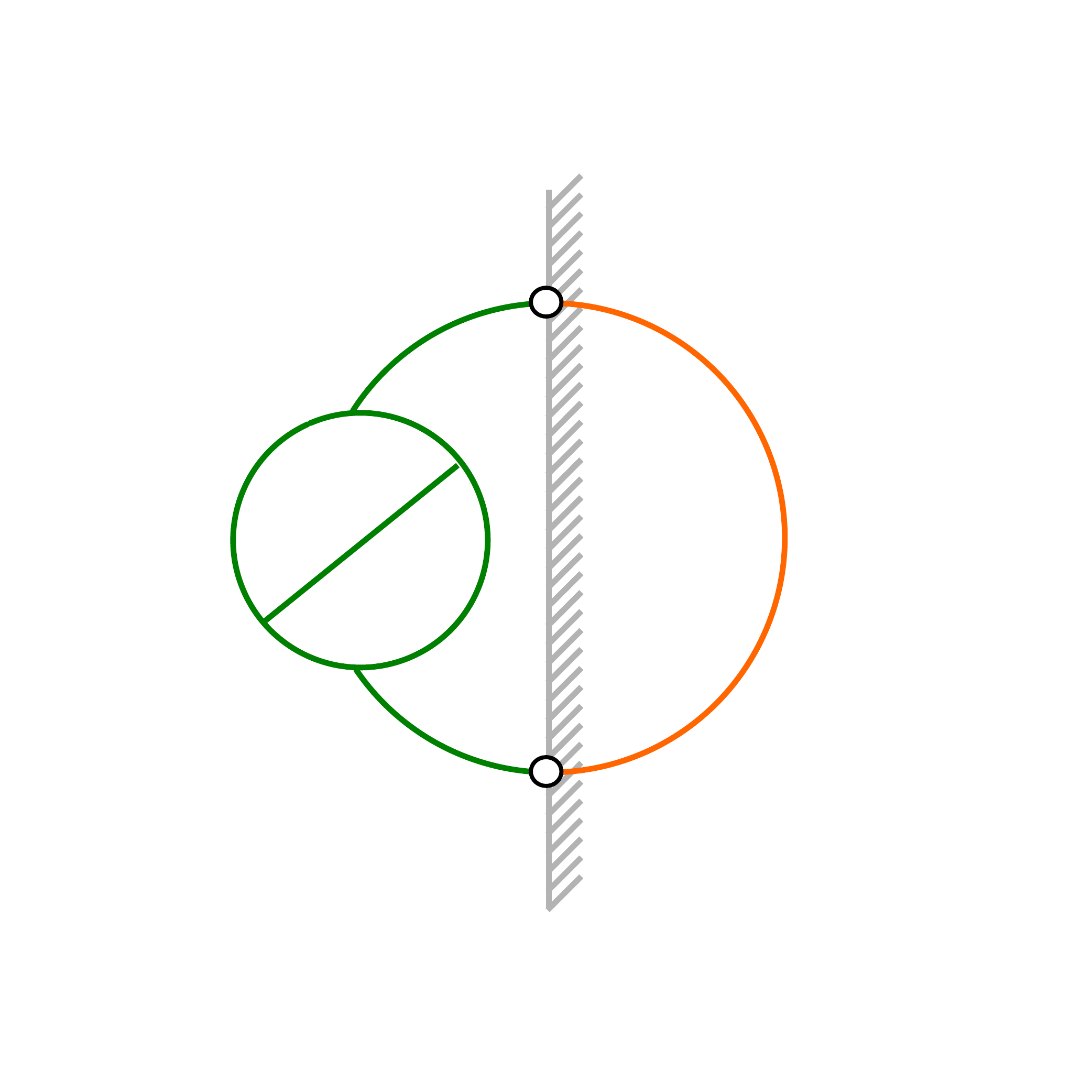}
        \caption{A bubble halfway absorbed into a symmetric boundary condition drawn for $D = 2$. The bulk $G_\cB$-foam is drawn in green. Where it ends on the boundary (white circles) there is a possible gravitational anomaly, described by a function $\beta_1: G_\cB \to \Omega^1$, which in a spin theory could be a fermionic operator, $\Omega^1_{spin} = \bZ_2$. This gravitational anomaly means that the $G_\cB$-foam with boundary is not isotopy invariant, and so we cannot apply the argument we gave in Section \ref{sec:defect} for Prop. \ref{propwog}. However, by introducing the subliminal $G_\cB$-foam (orange) carrying the invertible phase corresponding to $\beta_1$ (for $\Omega^1_{spin}$ it is a fermionic world line) the gravitational anomaly is cured and the entire object is isotopy invariant.}
        \label{fig:subliminal}
    \end{figure}

    Let us consider the next nontrivial example $k = 1$, which is associated with $\Omega^1_{Spin} = \bZ_2$, corresponding to a 0+1D invertible phase which is merely a state with odd fermion parity (there is no such phase for bosons as $\Omega^1_{SO} = 0$). These anomalies occur at the boundary of the so-called Gu-Wen supercohomology phases \cite{supercohomology}.
    
    Fermion parity is an unbreakable symmetry in any local fermionic system. Indeed, it is not possible for any charged order parameter to have long range order, since it anti-commutes with itself at separated points. This also applies when we have a boundary. However, we will see that at a $G_\cB$-symmetric boundary of a system with this anomaly, fermion parity cannot be conserved.
    
    The interpretation of the $k = 1$ anomaly is that $\alpha_1(g_1, \ldots, g_D) \in \bZ_2$ describes point-like junctions of $G_\cB$ domain walls that carry fermion parity. Conservation of fermion parity requires $\delta \alpha_1 = 0$ (as we expected from the AHSS). Otherwise we would be able to perform a recombination of the $G_\cB$-foam and change the total fermion parity of a correlation function. Note that here we use the fact that the height is $k=1$. If there were anomalous defects of higher dimension, it would be possible for fermions to be absorbed or emitted by them under these recombinations.
    
    Analogous to the previous section, if we have a symmetric boundary, it can absorb these junctions, disintegrating into a number of point-like boundary defects. Let us first suppose there are no gravitational anomalies of higher height that occur along the boundary. Then for fermion parity to be conserved, these boundary defects must carry an odd total fermion parity. The assignment of fermion parity to the defects is a function $\beta_1: G_\cB^{\times 2n-1} \to \bZ_2$ and fermion parity conservation then amounts to $\delta \beta_1 = \alpha_1$, contradicting our assumption of a nontrivial bulk anomaly.
    
    Suppose more generally that there are also higher dimensional boundary defects with gravitational anomalies which an absorb these fermionic defects. These boundary anomalies define a set of functions $\beta_j$. Analogous to the $k = 0$ case and Fig.~\ref{fig:absorb}, we can use these $\beta_k$ to define an extension of the correlation functions to $G_\cB$-foams crossing an interface, where each anomalous boundary mode is associated with an invertible phase decorating the $G_\cB$-foam on the subliminal side. We find
    \[0 = \alpha_1 + \delta \beta_1 + \Delta^D_1(\beta_{>1}).\]
    Indeed, by pushing a compact bulk foam (which may contain fermionic defects) to the subliminal side of the boundary, we see the fermionic defects are entirely accounted for by the invertible phases on the $G_\cB$-foam on the other side. Thus, using them as decorations in the bulk removes all fermionic defects from the bulk foam. This contradicts our assumption that the height is $k=1$. Therefore, there is no symmetric boundary condition.

    For $k > 1$ the argument is exactly the same, except we are now concerned with the ``conservation" of extended objects which host special modes. More precisely, a generic $k-1$-dimensional junction of domain walls is described by a tuple $(g_1,\ldots, g_{D+1-k})$ and the anomaly says that this defect carries the boundary modes of the $k$-dimensional invertible phase $\alpha_k(g_1,\ldots,g_{D+1-k}) \in \Omega^k$. For example, for $k = 2$, the invertible phase generating $\Omega^2_{Spin} = \bZ_2$ for fermions is the Kitaev string, and its boundary hosts a single Majorana zero mode. For $k = 3$ we may have defects decorated by holomorphic CFTs, and so on with higher dimensional systems with gravitational anomalies.
    
    The conservation law tells us that these objects must always occupy (images of) closed submanifolds in spacetime, in other words, it is Conjecture \ref{conjgravbdy}. Indeed, for $k = 2$ a single Majorana mode has no Hilbert space, so we cannot let it terminate on a state, and for $k = 3$ a holomorphic CFT admits no boundary conditions because it is chiral \cite{Jensen:2017eof} (see also Section~\ref{sec:conformal}). This is equivalent to the cocycle condition $\delta\alpha_k = 0$.\footnote{Recall that the cocycle condition only necessarily applies to $\alpha_k$, where $k$ is the height, since it is possible for certain anomalous defects of lower dimension to end on other, bigger anomalous defects, modifying the cocycle condition to \eqref{generalizedcocycle}.}

    To conclude the argument, we observe that at a symmetric boundary, the $k-1$-dimensional junctions which carry these anomalous modes can break open, disintegrating into a number of $k-1$-dimensional boundary junctions. As before, these anomalous modes must be carried by the $k-1$-dimensional boundary defects or be absorbed into anomalous defects of higher dimension. The condition that these modes are conserved at the boundary defines a set of $\beta_j$ such that
    \[0 = \alpha_k + \delta \beta_k + \Delta^D_k(\beta_{>k}),\]
    contradicting the assumption that the height is $k$. Thus, there can be no symmetric boundary.

    \end{proof}

                                                                                                                                                                                           Above we have given justification for Conjecture \ref{conjgravbdy} for gravitational anomalies in small dimensions, enough to conclude the theorem in dimensions $D \le 4$. Unfortunately it is difficult to extend these arguments to all dimensions without an intrinsic characterization of gravitational anomalies (rather than defining them as the boundaries of invertible phases and risking making a circular ``boundary of a boundary" argument). For perturbative anomalies of Lorentz invariant systems we can use the methods in Section 
    \ref{sec:conformal}, but for global anomalies things are more subtle. One might for example study certain diffeomorphisms in the presence of a boundary to deduce the vanishing of correlation functions, similar to what we have done in the $k = 0$ case with our special bubbles. This can likely be done for the next known gravitational anomaly, which occurs for $k = 5$ in bosonic systems and is associated with a large diffeomorphism of $\mathbb{CP}^2$ (complex conjugation of the homogeneous coordinates). For fermionic systems the next known gravitational anomaly has $k = 7$ and is a perturbative anomaly, for which our results of Section \ref{sec:conformal} apply. However, it is not even clear that all the gravitational anomalies encoded in the cobordism classification are detectable by diffeomorphisms (i.e. that the corresponding $Spin$ and $SO$ bordism classes are represented by mapping tori). In fact it seems rather unlikely to be the case. We therefore need a better intrinsic characterization of global gravitational anomalies. We leave this interesting question to future work.
    
    We used unitarity implicitly above in formulating the anomaly as a class in generalized cohomology. Indeed, the main TQFT justification for the cobordism classification relies on it \cite{freed2019reflection}. In Section~\ref{sec:nuea} we discussed the non-unitary $bc$ system which has a gauge-gravity anomaly but also a symmetric boundary condition. This anomaly has a similar feature to the ones we have described here, in particular it says that spacetime curvature binds vector charge. We can attempt a similar no-go argument as above in this case, by considering a flat disc which is deformed into a round hemisphere. In this process, bulk vector charge is created, but by the Gauss-Bonnet theorem, it can be balanced by vector charge associated with the extrinsic curvature of the boundary (see around \eqref{gcrd}). Thus, we cannot conclude the anomaly is trivial, in accordance with what happens in the $bc$ CFT.
    
    \section{Anomaly In-Flow Revisited}
    \label{sec:ainf}
    We did not use anomaly in-flow in the above arguments, but in fact a form of anomaly in-flow follows from our axioms. The basic idea is to draw the ``world-volume" of isotopies and recombinations of a $G$-foam on spacetime $\cM$ as a $G$-foam in $\cM \times [0,1]_s$. The flatness Axiom~\ref{ax:0} implies the same for this foam. Bubble creation/annihilation in Axiom~\ref{ax:2} corresponds to critical points of $s$ in the foam of Morse index 0 or $D+1$. The phase factors in Axioms~\ref{ax:3} and \ref{ax:4} define a phase $e^{2\pi i\Omega}$ associated with the foam, given by the product of $e^{2\pi i \omega}$ for each singularity. This phase is evidently local in the sense that the recombination phases (Axiom~\ref{ax:4}) are associated with the 0-dimensional singularities of this foam and the isotopy phases (Axiom~\ref{ax:3}) can be expressed as an integrated density on the defects.
    
    Lemma~\ref{lem:WZ} is equivalent to isotopy and recombination \emph{invariance} of this phase $e^{2\pi i\Omega}$ while holding the boundary $G$-foams fixed. Indeed, suppose there were two $G$-foams $A$ and $B$ on $\cM \times [0,1]_s$ which are related by isotopies and recombinations with phase factors $e^{2\pi i\Omega(A)}$ and $e^{2\pi i \Omega(B)}$ respectively. Since $A$ and $B$ have the same boundary foam, we can glue $A$ to an $s$-reversed copy of $B$ to obtain a foam on $\cM  \times [0,2]$ which has the same boundary on $\cM \times \{0\}$ as $\cM  \times \{2\}$ and a nontrivial phase factor $e^{2\pi i (\Omega(A) -  \Omega(B))}$ (here we use locality so the phases multiply). Because the two boundaries of this glued foam are the same, it represents a series of isotopies and recombinations which take the foam back to itself. By Lemma~\ref{lem:WZ}, $e^{2\pi i (\Omega(A) -  \Omega(B))} = 1$.
    
    The phase associated with this auxiliary foam is usually considered to define a topological action of a background gauge field in $D+1$-dimensions, which might be obtained by integrating out some gapped matter with an anomaly-free $G$-symmetry. If we define our anomalous theory $\cT$ on the boundary of one with the topological action $e^{-2\pi i\Omega}$, meaning we consider correlation functions of boundary operators in the presence of a bulk $G$-foam extending the boundary $G$-foam by defining the correlation function as the product of the boundary correlation function and $e^{-2\pi i\Omega}$ of the bulk foam, then by construction these correlation functions satisfy the desired axioms and are moreover isotopy and recombination invariant. Thus we say the bulk has cancelled the anomaly.
    
 \section{Discussion}
 \label{sec:discussion}

Let us comment on some applications and extensions of our obstruction theorems for symmetric boundary conditions in the presence of bulk 't Hooft anomalies. 

We have already explained in Section~\ref{sec:perteg} how our results are consistent with known boundary conditions of QFTs in the case of perturbative anomalies.  
As a simple application for non-perturbative anomalies, we can consider free non-chiral fermions in $D=1+1$, which have an anomalous $\mZ_2^L\times \mZ_2^R$ chiral fermion parity symmetry. The anomalies are captured by the $D=2+1$ SPT associated to $\Omega^3_{Spin}(B\mZ_2)=\mZ_8$ and the theory is anomaly free if the number of the Majorana fermions satisfies $N_f\in 8\mZ$ \cite{Fidkowski_2010,Ryu_2012,Qi_2013,Yao_2013}. Therefore $\mZ_2^L\times \mZ_2^R$ preserving boundary conditions are only possible for $N_f\in 8\mZ$. Recently conformal boundary conditions for an even number of Majorana fermions (with $N_f=2N$) preserving an anomaly free $U(1)^N$ global symmetry have been classified in \cite{Smith:2020nuf}, and indeed it was found that the chiral parity is preserved if and only if $N\in 4\mZ$.

 The arguments we have described readily extend to higher form symmetries \cite{Gaiotto_2015}. The definition of a $G$-foam is more or less the same except that there can be defects of codimension greater than one which are not junctions of hypersurfaces. For example, a 1-form symmetry will be associated with symmetry defects of codimension two. Anomalies are again described as either phases associated with recombinations of the foam or gravitational anomalies along certain junctions. The former is described by the ordinary cohomology of the classifying space of the higher form symmetry $n$-group, and the latter by a generalized cohomology thereof, most likely spin or oriented cobordisms.
 
 Our results also apply to domain walls between theories each with a different 't Hooft anomaly. A simple way to see this is to realize that any domain wall defines a boundary condition of the ``folded" theory, and if the two sides of the domain wall have different anomalies, the folded theory will be anomalous as well, and our results apply.
 
 This can be used to derive interesting properties of domain walls. For example, it is known that the $D=4$ Yang-Mills theories with a (1-form) center symmetry, such as with $SU(2)$ gauge group and adjoint matter, have a mixed anomaly between time reversal and the center symmetry  at $\theta = \pi$ \cite{thetatimereversaltemp}, while meanwhile at $\theta = 0$ there is no anomaly. Our results imply that at a domain wall between the two, there must be symmetry breaking. In particular, if time reversal is preserved, we must have center symmetry breaking, hence deconfinement on the wall. This differs from other mechanisms of deconfinement on domain walls studied in \cite{Sulejmanpasic_2017}, in which the domain wall itself is formed by breaking one of the anomalous symmetries (see also \cite{Hason_2020}).
 
 It may happen that there is some (normal) subgroup $H$ of the global symmetry $G$ such that all $H$-charged states are gapped, and we have an effective field theory of the low energy degrees of freedom with global symmetry $G/H$ (that acts faithfully). It is possible in this case that the $G/H$ symmetry of the effective field theory is anomalous, while the ``microscopic" $G$-symmetry is anomaly-free. This is called an emergent anomaly. In \cite{thorngren2020intrinsically}, the authors (including one of us) proposed that systems with emergent anomalies have certain SPT-like properties, including boundary conditions which must either break the symmetry or carry localized modes. Our results imply that if the $G/H$-symmetric effective field theory can describe the boundary condition, then there must be $G/H$ symmetry breaking. There are a few situations where this might not be the case. For example, if the gap to the $H$ charged states goes to zero at the boundary, we will find some gapless edge modes. On the other hand, there could be $H$ symmetry breaking or $H$-enriched topological order at the boundary.

\section*{Acknowledgements}
We thank Amos Yarom for useful discussions and correspondences, and Ruben Verresen for discussion and collaboration on related work. 
The work of YW is  supported in part by the Center for Mathematical Sciences and Applications and the Center for the Fundamental Laws of Nature at Harvard University.

\bibliographystyle{JHEP}
\bibliography{refs}

\begin{thebibliography}{10}
\bibitem{mathoverflow}
T.~Campion, ``Is lie group cohomology determined by restriction to finite
  subgroups?.'' MathOverflow.
\end{thebibliography}

\end{document}